\documentclass[sigconf]{acmart}

\usepackage{booktabs} 
\usepackage{amsmath,amssymb,amsthm}
\usepackage{xspace}
\usepackage{subcaption}

\newif\ifkdd
\kddfalse

\newif\iffinal
\finaltrue

\newtheorem{remark}{Remark}

\newcommand{\oldnote}[2][]{}
\iffinal
\newcommand{\note}[2][]{}
\else
\newcommand{\note}[2][]{{\bf [#1: {\it #2}]}}
\fi

\newcommand{\ptitle}[1]{\vspace{1mm}\noindent{\bf #1.}}

\newcommand{\cA}{\ensuremath{\mathcal{A}}}
\newcommand{\cW}{\ensuremath{\mathcal{W}}}
\newcommand{\tcW}{\ensuremath{\widetilde{\mathcal{W}}}}
\newcommand{\cT}{\ensuremath{\mathcal{T}}}
\newcommand{\st}{\ensuremath{\mbox{such that }}}
\newcommand{\bw}{\ensuremath{\mathbf{w}}}
\newcommand{\tV}{\ensuremath{\widetilde{V}}}
\newcommand{\tE}{\ensuremath{\widetilde{E}}}
\newcommand{\bE}{\ensuremath{\bar{E}}}
\newcommand{\tG}{\ensuremath{\widetilde{G}}}
\newcommand{\NP}{\ensuremath{\mathbf{NP}}\xspace}
\newcommand{\NPhard}{{\NP}-hard\xspace}

\newcommand{\bigO}{\ensuremath{\mathcal{O}}\xspace}
\newcommand{\STCbinary}{STCbinary}

\usepackage{enumitem}
\setlist[itemize]{leftmargin=*}

\usepackage{tikz,pgfplots,pgfplotstable}
\usetikzlibrary{matrix,positioning,fit,shapes,arrows,shadows,calc}

\definecolor{yafaxiscolor}{rgb}{0.3, 0.3, 0.3}
\definecolor{yafcolor1}{rgb}{0.4, 0.165, 0.553}
\definecolor{yafcolor2}{rgb}{0.949, 0.482, 0.216}
\definecolor{yafcolor3}{rgb}{0.47, 0.549, 0.306}
\definecolor{yafcolor4}{rgb}{0.925, 0.165, 0.224}
\definecolor{yafcolor5}{rgb}{0.141, 0.345, 0.643}
\definecolor{yafcolor6}{rgb}{0.965, 0.933, 0.267}
\definecolor{yafcolor7}{rgb}{0.627, 0.118, 0.165}
\definecolor{yafcolor8}{rgb}{0.878, 0.475, 0.686}
\definecolor{yafcolor9}{rgb}{0.965, 0.733, 0.767}

\ifkdd
\setcopyright{rightsretained}
\else
\setcopyright{none}
\fi

\ifkdd
\else
\renewcommand\footnotetextcopyrightpermission[1]{}
\fi

\acmDOI{10.475/123_4}
\acmISBN{123-4567-24-567/08/06}
\acmConference[KDD'18]{ACM SIGKDD Conference}{August 2018}{London, United Kingdom}
\acmYear{2018}
\copyrightyear{2018}

\acmPrice{15.00}

\begin{document}
\title[From acquaintance to best friend forever]{From acquaintance to best friend forever:\\
robust and fine-grained inference of social-tie strengths}

\author{Florian Adriaens}
\affiliation{%
  \institution{Dept. of Electronics and Information Systems, IDLab, Ghent University}
}
\email{florian.adriaens@ugent.be}
\author{Tijl De Bie}
\affiliation{%
  \institution{Dept. of Electronics and Information Systems, IDLab, Ghent University}
}
\email{tijl.debie@ugent.be}
\author{Aristides Gionis}
\affiliation{%
  \institution{Dept. of Computer Science\\Aalto University}
}
\email{aristides.gionis@aalto.fi}
\author{Jefrey Lijffijt}
\affiliation{%
  \institution{Dept. of Electronics and Information Systems, IDLab, Ghent University}
}
\email{jefrey.lijffijt@ugent.be}
\author{Polina Rozenshtein}
\affiliation{%
  \institution{Dept. of Computer Science\\Aalto University}
}
\email{polina.rozenshtein@aalto.fi}

\renewcommand{\shortauthors}{Adriaens et al.}

\begin{abstract}
Social networks often provide only a binary perspective on social ties:
two individuals are either connected or not.
While sometimes external information can be used to infer the \emph{strength} of social ties,
access to such information may be restricted or impractical.

Sintos and Tsaparas (KDD 2014) first suggested to infer the strength of social ties from the topology of the network alone,
by leveraging the \emph{Strong Triadic Closure (STC)} property. 
The STC property states that if person $A$ has strong social ties with persons $B$ and $C$,
$B$ and $C$ must be connected to each other as well (whether with a weak or strong tie).
Sintos and Tsaparas exploited this property to formulate the inference of the strength of social ties as an 
\NP-hard optimization problem,
and proposed two approximation algorithms.

We refine and improve this line of work,
by developing a sequence of linear relaxations of the problem, 
which can be solved exactly in polynomial time.
Usefully, these relaxations infer more fine-grained levels of tie strength (beyond strong and weak),
which also allows to avoid making arbitrary strong/weak strength assignments when the network topology provides inconclusive evidence.
One of the relaxations simultaneously infers the presence of a limited number of STC violations.
An extensive theoretical analysis leads to two efficient algorithmic approaches.
Finally, our experimental results elucidate the strengths of the proposed approach,
and sheds new light on the validity of the STC property in practice.
\end{abstract}

%

\keywords{Strong Triadic Closure, strength of social ties, Linear Programming, convex relaxations, half-integrality}

\maketitle

\note[Aris]{If it is OK with everyone, I would not call Sintos and Tsaparas' paper
``landmark'' or ``seminal'' (20 citations).}

\note[Aris]{For final check: make sure that we use consistently either node or vertex.
One suggestion is to use ``node'' for the graph, and ``vertex'' for the polytope.}

\oldnote[Aris]{I would suggest simplifying the names of our problems, 
e.g., STC-max instead of Bin-STC-max, LP1 instead of LP1-STC, etc.}

\section{Introduction}\label{sec:intro}

\oldnote[Tijl]{Introduce problem of inferring the strength of social ties, and its importance. Mention Sintos and Tsaparas and some later work.}

Online social networks, such as Facebook, provide unique insights into the social fabric of our society.
They form an unprecedented resource to study social-science questions, such as
how information propagates on a social network, how friendships come and go, how echo chambers work, how conflicts arise, and much more.
\ifkdd\else
\note[Tijl]{Should add citations}
\fi

Yet, many social networks provide a black-and-white perspective on friendship:
they are modeled by unweighted graphs, 
with an edge connecting two nodes representing that two people are friends. 
Surely though, some friendships are stronger than others,
and clearly, in studying social phenomena understanding the strength of social ties can be critical.

Although in some cases detailed data are available
and can be used for inferring the strength of social ties, 
e.g., communication frequency between users,
or explicit declaration of relationship types, 
such information may not always be available.

\note[Aris]{Removed mention on privacy, as 
it may raise concerns for our work attacking (legitimate) privacy considerations.}

The question of whether the strength of social ties can be inferred \emph{from the structure of the social network alone},
the subject of the current paper, is therefore an important one.
\ifkdd\else
Before we discuss our specific contributions, however, let us provide some essential background on prior work on this topic.
\fi

\oldnote[Jef]{I would suggest not to use subsections but instead use boldface non-indented paragraph markers as defined in the $\backslash$ptitle command. That should save a lot of space.}

\ptitle{Background}
\oldnote[Tijl]{Introduce basic notation and explain the original STC optimization problem.}
An important line of research attempting to address the inference of the strength of social ties
is based on the \emph{strong triadic closure} (STC) property from sociology,
introduced by Georg Simmel in 1908 \cite{sim:08}.
To understand the STC property, consider an undirected network $G=(V,E)$, 
with $E\subseteq {V \choose 2}$.
Consider additionally a \emph{strength function} 
$w:E\rightarrow \{\texttt{weak},\texttt{strong}\}$ assigning a binary strength value to each edge.
A triple of connected nodes $i,j,k\in V$ is said to satisfy the STC property, 
with respect to the strength function $w$,
if $w(\{i,j\})=w(\{i,k\}) = \texttt{strong}$ implies $\{j,k\}\in E$. 
In other words, two adjacent strong edges always need to be closed by an edge 
(whether weak or strong).
We refer to a strength function for which all connected triples satisfy the STC property as \emph{STC-compliant}:
\begin{definition}[STC-compliant strength function on a network]
A \emph{strength function} $w:E\rightarrow \{\texttt{weak},\texttt{strong}\}$ is STC-compliant
on an undirected network $G=(V,E)$ if and only if
\begin{align*}
&\mbox{for all } i,j,k\in V, \{i,j\},\{i,k\}\in E :\\
&w(\{i,j\})=w(\{i,k\})=\texttt{strong} \, \mbox{ implies } \{j,k\}\in E.
\end{align*}
\end{definition}
A consequence of this definition is that for an STC-compliant strength function,
any \emph{wedge}---defined as a triple of nodes $i,j,k\in V$ 
for which $\{i,j\},\{i,k\}\in E$ but $\{j,k\}\not\in E$---can include only one strong edge.
We will denote such a wedge by the pair $(i,\{j,k\})$, where $i$ is the root and $\{j,k\}$ are the end-points of the wedge,
and denote the set of wedges in a given network by~$\cW$.

On the other hand, for 
a \emph{triangle}---defined as a triple of nodes $i,j,k\in V$ for which
$\{i,j\},\{i,k\},\{j,k\}\in E$---no constraints are implied on the strengths of the three involved edges.
We will denote a triangle simply by the (unordered) set of its three nodes $\{i,j,k\}$,
and the set of all triangles in a given network as $\cT$.
\oldnote[Florian]{Technical detail: this definition implies that if $(i,j,k) \in \cW$, 
then also $(i,k,j) \in \cW$. 
Meaning we have two triples to denote the same wedge, something we'd rather not want? 
(two equal constraints in the LP)} 
\oldnote[Aris]{One way to address Florian's comment is to say that 
when we write $(i,j,k)$ the first item ($i$) is ordered, 
while the other two items ($j$ and $k$) are unordered.
I do not know how elegant is this, though.
We could also denote wedges as $(i,\{j,k\})$ or $\langle i,\{j,k\}\rangle$.}

Relying on the STC property,
Sintos and Tsaparas~\cite{sintos2014using} 
propose an approach to infer the strength of social ties.
They observe that a strength function that labels all edges as weak is always STC-compliant. 
However, as a large number of strong ties is expected to be found in a social network, 
they suggest  
searching for a strength function that
maximizes the number of strong edges, or (equivalently) minimizes the number of weak edges.

To write this formally,
we introduce a variable $w_{ij}$ 
for each edge $\{i,j\}\in E$,
defined as $w_{ij}=0$ if $w(\{i,j\})=\texttt{weak}$
and $w_{ij}=1$ if $w(\{i,j\})=\texttt{strong}$.
Then, the original STC problem, maximizing the number of strong edges, can be formulated as:
\begin{align}\tag{STCmax}\label{eq:Binary-STC-max}
\max_{w_{ij}:\{i,j\}\in E} & \sum_{\{i,j\}\in E} w_{ij},&&\\
\st & w_{ij} + w_{ik} \leq 1, &\mbox{ for all }& (i,\{j,k\})\in\cW,\label{eq:triangle-max}\\
& w_{ij}\in\{0,1\}, &\mbox{ for all }& \{i,j\}\in E.\label{eq:binary-max}
\end{align}
Equivalently, one could instead minimize $\sum_{\{i,j\}\in E} (1-w_{ij})$ subject to the same constraints,
or with transformed variables $v_{ij}=1-w_{ij}$ equal to $1$ for weak edges and $0$ for strong edges:
\begin{align}\tag{STCmin}\label{eq:Binary-STC-min}
\min_{v_{ij}:\{i,j\}\in E}& \sum_{\{i,j\}\in E} v_{ij},&&\\
\st & v_{ij} + v_{ik} \geq 1, &\mbox{ for all }& (i,\{j,k\})\in\cW,\label{eq:triangle-min}\\
& v_{ij}\in\{0,1\}, &\mbox{ for all }& \{i,j\}\in E.\label{eq:binary-min}
\end{align}
When we do not wish to distinguish between the two formulations,
we will refer to them jointly as \STCbinary.

Sintos and Tsaparas~\cite{sintos2014using} observe that \ref{eq:Binary-STC-min} is equivalent to Vertex Cover
on the so-called \emph{wedge graph} $G_E=(E,F)$, 
whose nodes are the edges of the original input graph $G$, 
and whose edges are $F=\{(\{i,j\},\{i,k\}) \mid (i,\{j,k\})\in\cW\}$,
i.e., two nodes of $G_E$ are connected by an edge if the edges they represented in $G$ form a wedge.
While Vertex Cover is \NPhard,
a simple factor-$2$ approximation algorithm 
can be adopted for \ref{eq:Binary-STC-min}. 
On the other hand, \ref{eq:Binary-STC-max}
is equivalent to finding the {\em maximum independent set} on the wedge graph $G_E$, 
or equivalently the {\em maximum clique} on the {\em complement} of the wedge graph.
It is known that there cannot be a polynomial-time algorithm that for every real number
$\varepsilon > 0$ approximates the maximum clique to within a factor better than
${\mathcal O}(n^{1-\varepsilon})$~\cite{haastad1999clique}.
In other words, while a polynomial-time approximation algorithm exists for minimizing the number of weak edges (with approximation factor two),
no such polynomial-time approximation algorithm exists for maximizing the number of strong edges.

Despite its novelty and elegance, 
\STCbinary\ 
suffers from a number of weaknesses,
which we address in this paper. 

First, \STCbinary\ is an \NPhard problem.
Thus, one has to either resort to approximation algorithms, 
which are applicable only for certain problem variants---see
the discussion on \ref{eq:Binary-STC-min} vs.~\ref{eq:Binary-STC-max} above---or 
rely on exponential algorithms and hope for good behavior in practice. 
Second, the problem returns \emph{only binary edge strengths}, 
weak vs. strong. 
In contrast, real-world social networks contain
tie strengths of many different levels.
A third limitation is that, on real-life networks,
\STCbinary\ tends to have many optimal solutions.
Thus, any such optimal solution makes \emph{arbitrary strength assignments}
for the edges where different optimal solutions differ from each other.%
\footnote{A case in point is a star graph, where the optimal solution contains one strong edge (arbitrarily selected),
while all others are weak.}
Last but not least, \STCbinary\
\emph{assumes that the STC property holds for all wedges}. 
Yet, real-world social networks tend to be noisy,
with spurious connections as well as missing edges.

\ptitle{Contributions}
In this paper we propose a 
series of linear programming relaxations that address all of the above limitations of \STCbinary. 
In particular, our LP relaxations provide the following advantages.
\begin{itemize}
\item 
The first relaxation replaces the integrality constraints $w_{ij}\in\{0,1\}$
with fractional counterparts $0\le w_{ij}\le 1$.
It can be shown that this relaxed LP is \emph{half-integral}, 
i.e., the edge strengths in the optimal solution take values $w_{ij}\in\{0,\frac{1}{2}, 1\}$.
Thus, not only the problem becomes polynomial, 
but the formulation introduces meaningful \emph{three-level} social strengths.
\item 
Next we relax the upper-bound constraint, requiring only $w_{ij}\ge 0$,
while generalizing the STC property to deal with higher gradations of edge strengths.
We show that the optimal edge strengths still
take values in a small discrete set\ifkdd\else
, controlled with an additional parameter. 
\fi 
Thus, our approach can yield multi-level edge strengths, 
from a small set of discrete values, 
while ensuring a polynomial algorithm. 
\item 
We show how the previous relaxations can be solved 
by advanced and highly efficient combinatorial algorithms, 
so that one need not rely on generic LP solvers.
\item
As our relaxations allow intermediate strength levels,
arbitrary choices between weak and strong values 
can be avoided by assigning an intermediate strength.
\ifkdd\else
Furthermore, the computational tractability of the relaxed solution
makes it possible to also quantify in polynomial time the 
\emph{range of possible strengths} an edge can have in the set of optimal strength assignments.
For \STCbinary, given the intractability of finding even one optimal solution,
this is clearly beyond reach.
\fi
\item 
Our final relaxation simultaneously edits the network
while optimizing the edge strengths,
making it robust against noise in the network.
Also this variant has no integrality constraints, 
and thus, it can again be solved in polynomial time.%
\footnote{Note that Sintos and Tsaparas~\cite{sintos2014using} also suggest a variant of
\STCbinary\ that allows the introduction of new edges.
However, the resulting problem is again NP-hard,
and the provided algorithm provides an $\bigO(\log(|E|))$-approximation 
rather than a constant-factor approximation.
\note[Tijl]{The paper itself mentions $\bigO(\log(n))$ -- I think $n=|E|$ but I could not immediately find it in the paper.}}
\end{itemize}

\ptitle{Outline}
We start by proposing the successive relaxations in Sec.~\ref{sec:relaxations}.
In Sec.~\ref{sec:analysis} we analyse these relaxations and derive properties of their optima,
highlighting the benefits of these relaxations with respect to \STCbinary.
The theory developed in Sec.~\ref{sec:analysis}
leads to efficient algorithms, discussed in Sec.~\ref{sec:algorithms}.
Empirical performance is evaluated in Sec.~\ref{sec:empirical} and related work is
reviewed in Sec.~\ref{sec:relatedwork},
before drawing conclusions in Sec.~\ref{sec:conclusions}.

\section{LP relaxations}\label{sec:relaxations}

Here we will derive a sequence of increasingly loose relaxations 
of Problem~\ref{eq:Binary-STC-max}.
Their detailed analysis is deferred to Sec.~\ref{sec:analysis}.%
\footnote{Our relaxations can also be applied to Problem~\ref{eq:Binary-STC-min}, 
however, for brevity, hereinafter we omit discussion on this minimization problem.}

\subsection{Elementary relaxations}

In this subsection we simply enlarge the feasible set of strengths $w_{ij}$,
for all edges $\{i,j\}\in E$.
This is done in two steps.


\subsubsection{Relaxing the integrality constraint}

The first relaxation relaxes the  
constraint $w_{ij}\in\{0,1\}$ to 
$0\leq w_{ij}\leq 1$.
Denoting the set of edge strengths with $\bw=\{w_{ij}\mid\{i,j\}\in E\}$, 
this yields:
\begin{align}\tag{LP1}\label{eq:LP1-STC}
\max_{\bw} & \sum_{\{i,j\}\in E} w_{ij},&&\\
\st & w_{ij} + w_{ik} \leq 1, &\mbox{ for all }& (i,\{j,k\})\in\cW,\label{eq:wedge}\\
& w_{ij}\geq 0, &\mbox{ for all }& \{i,j\}\in E, \label{eq:lower}\\
& w_{ij}\leq 1, &\mbox{ for all }& \{i,j\}\in E.\label{eq:upper}
\end{align}
\ifkdd\else
Equivalently in Problem~\ref{eq:Binary-STC-min} one can relax constraint~(\ref{eq:binary-min}) to $0\leq v_{ij}\leq 1$.
Recall that Problems~\ref{eq:Binary-STC-max} and \ref{eq:Binary-STC-min} are equivalent respectively with the Independent Set and Vertex Cover problems on the wedge graph.
For those problems, this particular linear relaxation is well-known, and for Vertex Cover it can be used to achieve a 2-approximation \cite{Hoc:82,Hoc:83}.
\oldnote[Aris]{The previous paragraph can be easily omitted, if we need space.}
\fi 

Clearly, this relaxation will lead to solutions that are not necessarily binary.
However, as will be explained in Sec.~\ref{sec:analysis}, 
Problem~\ref{eq:LP1-STC} is \emph{half-integral},
meaning that there always exists an optimal solution with values 
$w_{ij}\in\{0,\frac 12,1\}$ for all $\{i,j\}\in E$.

\subsubsection{Relaxing the upper bound constraints to triangle constraints}

We now further relax Problem~\ref{eq:LP1-STC},
so as to allow for edge strengths larger than~$1$. 
The motivation to do so is to allow for higher gradations in the inference of edge strengths.

Simply dropping the upper-bound constraint~(\ref{eq:upper}) 
would yield uninformative unbounded solutions, 
as edges that are not part of any wedge would be unconstrained.
Thus, the upper-bound constraints cannot simply be deleted; 
they must be replaced by looser constraints
that bound the values of edge strengths in triangles
in the same spirit as the STC constraint does for edges in wedges.

To do so, we propose to generalize the wedge STC constraints~(\ref{eq:wedge})
to STC-like constraints on triangles, as follows: 
\emph{in every triangle, the combined strength of two adjacent edges
should be bounded by an increasing function of the strength of the closing edge}.
In social-network terms: the stronger a person's friendship with two other people,
the stronger the friendship between these two people must be.
Encoding this intuition as a linear constraint yields:
\begin{align*}
w_{ij}+w_{ik} \leq c + d\cdot w_{jk},
\end{align*}
for some $c,d\in\mathbb{R}^+$.
This is the most general linear constraint
that imposes a bound on $w_{ij}+w_{ik}$ that is increasing with $w_{jk}$,
as desired.
We will refer to such constraints as \emph{triangle constraints}.

\oldnote[Jef]{Although there is nothing wrong with this per se, the precise form of the constraint appears a bit out of nowhere. If possible, the motivation for its form could be extended.}

In sum, we relax Problem~\ref{eq:LP1-STC} by first adding the triangle constraints for all triangles,
and subsequently dropping the upper-bound constraints~(\ref{eq:upper}).
For the resulting optimization problem to be a \emph{relaxation} of Problem~\ref{eq:LP1-STC},
the triangle constraints must be satisfied throughout the original feasible region.
This is the case as long as $c\geq 2$:
indeed, then the box constraints $0\leq w_{ij}\leq 1$ ensure that the triangle constraint is always satisfied.
The tightest possible relaxation 
is thus achieved with $c=2$,
yielding the following relaxation:
\begin{align}\tag{LP2}\label{eq:LP2-STC}
\max_{\bw} & \sum_{\{i,j\}\in E} w_{ij},&&\\
\st & w_{ij} + w_{ik} \leq 1, &\mbox{ for all }& (i,\{j,k\})\in\cW,\nonumber\\
& w_{ij}+w_{ik} \leq 2 + d\cdot w_{jk}, &\mbox{ for all }& \{i,j,k\}\in\cT,\label{eq:triangle}\\
& w_{ij}\geq 0, &\mbox{ for all }& \{i,j\}\in E.\nonumber
\end{align}

\begin{remark}[The wedge constraint is a special case of the triangle constraint]\label{rem:absent}
Considering an absent edge as an edge with negative strength $-1/d$,
the wedge constraint can in fact be regarded as a special case of the triangle constraint.
\end{remark}

\oldnote[Tijl]{Introduce parameterized triangle inequality here, with variable $d$. Perhaps this can be done most directly by setting the strength of an absent edge to $-w$, a weak edge to $0$, and a strong edge to a parameter $1$. Then, the constraint is generally of the form $x+y<c+d*z$. To make sure the relaxation is as tight as possible, $c$ and $d$ should be as small as possible while containing the feasible region of the unrelaxed solution. Thus, it should be as tight as possible for all feasible discrete edge strength combinations. The two tightest constraints are for $x=0,y=1,z=-w$, and for $x=1,y=1,z=0$. This gives two equations in two variables, solved by setting $1=c-d*w$ and $2=c$, such that $c=2$ and $w=1/d$. It should be pointed out that this effectively means bounding $\frac{x+y}{d}$ by the closing edge strength $z$, plus a constant $c$. For $0<d<1$, we haven't really thought it through yet! It looks like then there might be more different levels than 6... To be investigated, or omitted from this paper?}

\subsection{Enhancing robustness by allowing edge additions and deletions}

As noted earlier, although the STC property is theoretically motivated, 
real-world social networks are noisy and may contain many exceptions to this rule.
In this subsection we propose two further relaxations of Problem~\ref{eq:LP2-STC}
that gracefully deal with exceptions of two kinds:
wedges where the sum of edges strengths exceeds $1$,
and edges with a negative edge strength,
indicating that the STC property would be satisfied should the edge not be present. 

These relaxations thus solve the STC problem while allowing a small number of edges to be added or removed from the network.
\note[Tijl]{Discuss relation with the similar task in Sintos and Tsaparas.
Perhaps we can explain things also better by connecting it directly to that work.}

\subsubsection{Allowing violated wedge STC constraints}

In order to allow for violated wedge STC constraints,
we can simply add positive \emph{slack variables} $\epsilon_{jk}$ for all $(i,\{j,k\})\in\cW$:
\begin{align}
\label{constraint:wedge-slack}
w_{ij}+w_{ik}\leq 1+\epsilon_{jk},\quad
\epsilon_{jk}\geq 0. 
\end{align}
Elegantly, the slacks $\epsilon_{jk}$ can be interpreted as quantifying the strength of the (absent) edge between $j$ and $k$.
To show this,
let $\bE$ denote the set of pairs of end-points of all the wedges in the graph, 
i.e., $\bE = \{ \{j,k\} \mid \mbox{there exists } i\in V:(i,\{j,k\})\in \cW \}$.
We also extend our notation to introduce strength values for 
those pairs, i.e., 
$\bw=\left\{w_{ij}\mid\{i,j\}\in E \mbox{ or } \{i,j\}\in \bE\right\}$,
and define $w_{jk}\triangleq\frac{\epsilon_{jk}-1}{d}$ for $\{j,k\}\in\bE$.
The relaxed wedge constraints~(\ref{constraint:wedge-slack})
are then formally identical to the triangle STC constraints~(\ref{eq:triangle}).
Meanwhile, the lower bound $\epsilon_{jk}\geq 0$ from~(\ref{constraint:wedge-slack})
implies $w_{jk}\geq -\frac{1}{d}$,
i.e., allowing the strength of these absent edges to be negative.

In order to bias the solution towards few violated wedge constraints
a term $-C\sum_{\{j,k\}\in\bE}w_{jk}$ is added to the objective function.
The larger the parameter $C$, the more a violation of a wedge constraint will be penalized.
The resulting problem is:
\begin{align}\tag{LP3}\label{eq:LP3-STC}
\max_{\bw} & \sum_{\{i,j\}\in E} w_{ij} - C\sum_{\{j,k\}\in\bE}w_{jk},&&\\
\st & w_{ij} + w_{ik} \leq 2+d\cdot w_{jk}, &\mbox{ for all }& (i,\{j,k\})\in\cW,\nonumber\\
& w_{ij}+w_{ik} \leq 2 + d\cdot w_{jk}, &\mbox{ for all }& \{i,j,k\}\in\cT,\nonumber\\
& w_{ij}\geq 0, &\mbox{ for all }& \{i,j\}\in E.\nonumber\\
& w_{jk}\geq -\frac{1}{d}, &\mbox{ for all }& \{j,k\}\in\bE.\nonumber
\end{align}

Note that in Remark~\ref{rem:absent},
$-\frac{1}{d}$ was argued to correspond to the strength of an absent edge.
Thus, the lower-bound constraint on $w_{jk}$ 
requires these weights to be at least as large as the weight that signifies an absent edge.
If it is strictly larger, this may suggest that the edge is in fact missing,
as adding it increases the sum of strengths in the objective more than the penalty paid for adding it.

\ifkdd\else
\note[Tijl]{Add result on relation between $C$ and the number of violated wedge constraints?}
\fi

\subsubsection{Allowing negative edge strengths}\label{sec:negedge}

The final relaxation is obtained by allowing edges to have negative strength,
with lower bound equal to the strength signifying an absent edge:
\begin{align}\tag{LP4}\label{eq:LP4-STC}
\max_{\bw} & \sum_{\{i,j\}\in E} w_{ij} - C\sum_{\{j,k\}\in\bE}w_{jk},&&\\
\st & w_{ij} + w_{ik} \leq 2+d\cdot w_{jk}, &\mbox{ for all }& (i,\{j,k\})\in\cW,\nonumber\\
& w_{ij}+w_{ik} \leq 2 + d\cdot w_{jk}, &\mbox{ for all }& \{i,j,k\}\in\cT,\nonumber\\
& w_{ij}\geq -\frac{1}{d}, &\mbox{ for all }& \{i,j\}\in E.\nonumber\\
& w_{jk}\geq -\frac{1}{d}, &\mbox{ for all }& \{j,k\}\in \bE.\nonumber
\end{align}

This formulation allows the optimization problem to strategically delete some edges from the graph,
if doing so allows it to increase the sum of all edge strengths.

\note[Jef]{I had not realized this yet, but the definition is asymmetric, in that the cost to remove an existing edge is low as compared to inferring it. Due to linearity, I guess we cannot do anything about it but maybe a comment needs to be inserted to discuss this?}
\note[Florian]{Jef, Adding an extra cost term $C_{add}$ for deleting edges would just reduce it, after dividing by $1-C_{add}$, to the same problem right?  In that way, edge deletion cost is incorporated in the current meaning of C.}
\note[Tijl]{My understanding is that the relative cost of adding and deleting is determined by $D$.
This is probably fine, no? Why should they incur the same cost?}

\section{Theoretical analysis of the optima}\label{sec:analysis}

The general form of relaxation ~\ref{eq:LP1-STC} is a well-studied problem, 
and it is known that there always exists a half-integral solution---a solution
where all $w_{ij}\in\{0,\frac{1}{2},1\}$ \cite{NeT:75}.
In this section we demonstrate and exploit the existence of symmetries in the optima to show an analogous result for Problem~\ref{eq:LP2-STC}.
Furthermore, the described symmetries also exist for Problems~\ref{eq:LP3-STC} and \ref{eq:LP4-STC},
although they do not imply an analogue of the half-integrality result for these problems.

We also discuss how the described symmetries are useful in reducing the arbitrariness of the optima,
as compared to Problems~\ref{eq:Binary-STC-max} and \ref{eq:Binary-STC-min},
where structurally-indistinguishable edges might be assigned different strengths at the optima.
Furthermore, in Sec.~\ref{sec:algorithms} we will show how the symmetries can be exploited for algorithmic performance gains, as well.

We start by giving some useful definitions and lemmas.

\subsection{Auxiliary definitions and results}\label{sec:definitions}

\begin{figure}[t]
\begin{tikzpicture}

\tikzstyle{exedge} = [yafcolor5!80, thick, text=black!80]
\tikzstyle{exnode} = [thick, draw = yafcolor7!80, fill=white, circle, inner sep = 1pt, text=black, minimum width=11pt]

\node[exnode] (x) at (0, 0.5)    {$x$};
\node[exnode] (y) at (0.95, 0.2) {$y$};
\node[exnode] (z) at (0.6, 1.3)  {$z$};
\node[exnode] (w) at (1.55, 1.1) {$w$};
\node[exnode] (u) at (2.5, 0.9)  {$u$};

\draw[-, exedge, bend left = 10] (w) to (y);
\draw[-, exedge, bend left = 10] (y) to (x);
\draw[-, exedge, bend left = 10] (x) to (z);
\draw[-, exedge, bend left = 10] (z) to (w);
\draw[-, exedge, bend left = 10] (w) to (x);
\draw[-, exedge, bend left = 10] (z) to (y);
\draw[-, exedge, bend left = 10] (w) to (u);

\end{tikzpicture}
\caption{\label{figure:edgetypes}A toy graph illustrating 
the different type of edges defined in Section~\ref{sec:definitions}.}
\end{figure}
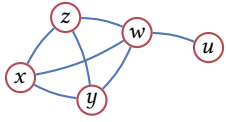

It is useful to distinguish two types of edges:
\begin{definition}[Triangle edge and wedge edge]
A \emph{triangle edge} is an edge that is part of at least one triangle, but that is part of no wedge.
A \emph{wedge edge} is an edge that is part of at least one wedge.
\end{definition}
These definitions are illustrated in a toy graph in Figure~\ref{figure:edgetypes},
where edges $(x,y)$, $(y,z)$, and $(x,z)$ are triangle edges, 
while edges $(w,x)$, $(w,y)$, $(w,z)$, and $(w,u)$ are wedge edges.

It is clear that in this toy example the set of triangle edges forms a clique.
This is in fact a general property of triangle edges:
\begin{lemma}[Subgraph induced by triangle edges]\label{lem:triangleclique}
Each connected component in the edge-induced subgraph, induced by all triangle edges, is a clique.
\end{lemma}
\ifkdd\else
\begin{proof}
The contrary would imply the presence of a wedge within this subgraph,
which is a contradiction since by definition none of the triangle-edges can be part of any wedge.
\end{proof}
\fi
Thus, we can introduce the notion of a \emph{triangle clique}. 
\begin{definition}[Triangle cliques]
The connected components 
in the edge-induced subgraph induced by all triangle edges are called \emph{triangle cliques}.
\end{definition}

The nodes $\{x,y,z\}$ 
in Figure~\ref{figure:edgetypes}
form a triangle clique.
Note that not every clique in a graph is a triangle clique.
E.g., 
nodes $\{x,y,z,w\}$ form a clique but not a triangle clique. 

A node $k$ is a \emph{neighbor} of a triangle clique $C$ 
if $k$ is connected to at least one node of $C$.
It turns out that a neighbor of a triangle clique
is connected to all the nodes of that triangle clique.

\begin{lemma}[Neighbors of a triangle clique]\label{lem:neighbors}
Consider a triangle clique $C\subseteq V$,
and a node $k\in V\setminus C$.
Then, 
either $\{k,i\}\not\in E$ for all $i\in C$, 
or $\{k,i\}\in E$ for all $i\in C$.
\end{lemma}
In other words, a neighbor of one node in the triangle clique must be a neighbor of them all,
in which case we can call it a \emph{neighbor of the triangle clique}.
\ifkdd\else
\begin{proof}
Assume the contrary, that is, 
that some node $k\in V\setminus C$ is connected to $i\in C$ but not to $j\in C$.
This means that $(i,\{j,k\})\in\cW$.
However, this contradicts the fact that $\{i,j\}$ is triangle edge.
\end{proof}
\fi
This lemma allows us to define the concepts \emph{bundle} and \emph{ray}:
\begin{definition}[Bundle and ray]
Consider a triangle clique $C\subseteq V$ and one of its neighbors $k\in V\setminus C$.
The set of edges $\{k,i\}$ connecting $k$ with $i\in C$ is called a \emph{bundle} of the triangle clique.
Each edge $\{k,i\}$ in a bundle is called a \emph{ray} of the triangle clique.
\end{definition}

In Figure~\ref{figure:edgetypes}
the edges $(w,x)$, $(w,y)$, and $(w,z)$ form a bundle
of the triangle clique with nodes $x,y,$ and $z$.

\paragraph{A technical condition to ensure finiteness of the optimal solution}
Without 
loss of generality,
we will further assume that no connected component of the graph is a clique ---
such connected components can be easily detected and handled separately.
This ensures that a finite optimal solution exists, 
as we show in Propositions~\ref{prop:finite-noslacks} and~\ref{prop:finite-slacks}.
\note[Tijl]{To be checked: do we need the next lemma (saying that each triangle edge is adjacent to a wedge edge---omitted in the KDD version) for anything else?}
\ifkdd\else
These propositions rest on the following lemma:
\begin{lemma}[Each triangle edge is adjacent to a wedge edge]\label{lem:triangle-next-to-wedge-edge}
Each triangle edge in a graph without cliques as connected components is immediately adjacent to a wedge edge.
\end{lemma}
\fi
\ifkdd\else
\begin{proof}
First note that, in a connected component,
for each edge there exists at least one adjacent edge.
If a triangle edge $\{i,j\}$ were adjacent to triangle edges only,
all $i$'s and $j$'s neighbors would be connected.
Together with $i$ and $j$, this set of neighbors would form a connected component 
that is a clique---a contradiction.
\end{proof}
\fi
\begin{proposition}[Finite feasible region without slacks]\label{prop:finite-noslacks}
A graph in which no connected component is a clique has a finite feasible region
for Problems~\ref{eq:LP1-STC} and \ref{eq:LP2-STC}.
\end{proposition}
Thus, also the optimal solution is finite.
\ifkdd\else
\begin{proof}
The weight of wedge edges is trivially bounded by $1$.
From Lemma~\ref{lem:triangle-next-to-wedge-edge}, we know that
the weight of each triangle edge is bounded by a at least one triangle inequality where the strength of the edge on the right hand side
is a for a wedge edge---i.e., it is also bounded by a finite number,
thus proving the theorem.
\end{proof}
\fi
For Problems~\ref{eq:LP3-STC} and \ref{eq:LP4-STC} the following weaker result holds:
\begin{proposition}[Finite optimal solution with slacks]\label{prop:finite-slacks}
A graph in which no connected component is a clique has a finite optimal solution
for Problems~\ref{eq:LP3-STC} and \ref{eq:LP4-STC} for sufficiently large $C$.
\end{proposition}
Note that for these problems the feasible region is unbounded.
\ifkdd\else
\begin{proof}
Let $n$ be the number of nodes in the largest connected component in the graph.
Let $j,k\in V$ be two nodes in this connected component for which $\{j,k\}\not\in E$.
We will first show that $w_{jk}$ is finite at the optimum for sufficiently large $C$,
by showing that increasing it from a finite value strictly and monotonously reduces the value of the objective function.

By increasing $w_{jk}$ by $\delta$,
the other weights in the connected component may each increase by at most $\max\{d,d^2\}\cdot\delta$.
Indeed, the left hand side in the wedge constraints for wedges $(i,\{j,k\})$ can increase by $d\cdot\delta$,
which may result in a potential increase of the left hand sides of the triangle constraints by $d^2\cdot\delta$.
As there are at most $n^2$ edges in the connected component,
this means that the objective function will strictly decrease if $n^2\cdot\max\{d,d^2\}<C$,
as we set out to prove.

If $w_{jk}$ for all $(i,\{j,k\})\in\cW$ is finite,
this means that also each wedge edge is finitely bounded.

From Lemma~\ref{lem:triangle-next-to-wedge-edge},
it thus also follows that each triangle edge is finitely bounded.
\end{proof}
\fi

\subsection{Symmetry in the optimal solutions}

We now proceed to show that certain symmetries exist in \emph{all} optimal solutions%
\ifkdd
\else
 (Sec.~\ref{sec:symm-all}),
\fi
while for other symmetries we show that there always \emph{exists} 
an optimal solution that exhibits it%
\ifkdd
\else
 (Sec.~\ref{sec:symm-some}).
\fi

\subsubsection{There always exists an optimal solution that exhibits symmetry}\label{sec:symm-some}

We first state a general result,
before stating a more practical corollary.
The theorem pertains to automorphisms $\alpha:V\rightarrow V$ of the graph $G$,
defined as node permutations that leave the edges of the graph unaltered:
for $\alpha$ to be a graph automorphism, it must hold that $\{i,j\}\in E$ if and only if $\{\alpha(i),\alpha(j)\}\in E$.
Graph automorphisms form a permutation group defined over the nodes of the graph.

\begin{theorem}[Invariance under graph automorphisms]\label{thm:automorphisms}
For any subgroup $\cA$ of the graph automorphism group of $G$,
there exists an optimal solution for Problems~\ref{eq:LP1-STC}, \ref{eq:LP2-STC}, \ref{eq:LP3-STC} and \ref{eq:LP4-STC}
that is invariant under all automorphisms $\alpha\in\cA$.
In other words, there exists an optimal solution $\bw$ such that $w_{ij}=w_{\alpha(i)\alpha(j)}$ for each automorphism $\alpha\in\cA$.
\end{theorem}
\ifkdd\else
\begin{proof}
Let $w_{ij}$ be the optimal strength for the node pair $\{i,j\}$
in an optimal solution $\bw$.
Then, we claim that assigning a strength $\frac{1}{|\cA|}\sum_{\alpha\in\cA} w_{\alpha(i)\alpha(j)}$ 
to each node pair $\{i,j\}$ is also an optimal solution.
This solution satisfies the condition in the theorem statement, so if true, the theorem is proven.

It is easy to see that this strength assignment has the same value of the objective function.
Thus, we only need to prove that it is also feasible.

As $\alpha$ is a graph automorphism, it preserves the presence of edges, wedges, and triangles (e.g., $\{i,j\}\in E$ if and only if $\{\alpha(i),\alpha(j)\}\in E$).
Thus, if a set of strengths $w_{ij}$ for node pairs $\{i,j\}$ is a feasible solution,
then also the set of strengths $w_{\alpha(i)\alpha(j)}$ is feasible for these node pairs.
Due to convexity of the constraints, also the average over all $\alpha$ of these strengths is feasible,
as required.
\end{proof}
\fi

Enumerating all automorphisms of a graph is computationally at least as hard as solving the graph-isomorphism problem.
The graph-isomorphism problem is known to belong to $\mathbf{NP}$, 
but it is not known whether it belongs to $\mathbf{P}$.
However, the set of permutations in the following proposition is easy to find.
\begin{proposition}
The set $\Pi$ of permutations $\alpha:V\rightarrow V$ for which $i\in C$ if and only if $\alpha(i)\in C$ for all triangle cliques $C$ in $G$ forms a subgroup of the automorphism group of $G$.
\end{proposition}
Thus the set $\Pi$ contains permutations of the nodes that map any node in a triangle clique onto another node in the same triangle clique.
\ifkdd\else
\begin{proof}
Each permutation $\alpha\in\Pi$ is an automorphism of $G$.
This follows directly from Lemma~\ref{lem:neighbors} and the fact that $\alpha$ only permutes nodes \emph{within} each triangle clique.
Furthermore, it is clear that if $\alpha\in\Pi$ then also $\alpha^{-1}\in\Pi$, and if $\alpha_1,\alpha_2\in\Pi$ then also $\alpha_1\alpha_2\in\Pi$.
Finally, $\Pi$ contains at least the identity and is thus non-empty, proving that $\Pi$ is a subgroup of $\cA$.
\end{proof}
\fi

We can now state the more practical Corollary of Theorem~\ref{thm:automorphisms}:
\begin{corollary}[Invariance under permutations within triangle cliques]\label{cor:permutations}
Let $\Pi$ be the set of permutations $\alpha:V\rightarrow V$ for which $i\in C$ if and only if $\alpha(i)\in C$ for all triangle cliques $C$. 
There exists an optimal solution $\bw$ for problems~\ref{eq:LP1-STC}, \ref{eq:LP2-STC}, \ref{eq:LP3-STC} and \ref{eq:LP4-STC}
for which $w_{ij}=w_{\alpha(i)\alpha(j)}$ for each permutation $\alpha\in\Pi$.
\end{corollary}
Thus there always exists an optimal solution for which edges in the same triangle clique (i.e., adjacent triangle edges) have equal strength,
and for which rays in the same bundle have equal strength.

\ifkdd\else
Such a symmetric optimal solution can be constructed from any other optimal solution,
by setting the strength of a triangle edge equal to the average of strengths within the triangle clique it is part of,
and setting the strength of each ray equal to the average of the strengths within the bundle it is part of.
Indeed, this averaged solution is equal to the average of all permutations of the optimal solution,
which, from convexity of the problem, is also feasible and optimal.
\fi

\subsubsection{In each optimum, connected triangle-edges have equal strength}\label{sec:symm-all}

\ifkdd
\else
Here, we will prove that only
\fi
some of the symmetries discussed above are present in \emph{all} optimal solutions,
as formalized by the following theorem:

\begin{theorem}[Optimal strengths of adjacent triangle edges are equal]\label{thm:all}
In any optimal solution of Problems~\ref{eq:LP1-STC}, \ref{eq:LP2-STC}, \ref{eq:LP3-STC} and \ref{eq:LP4-STC},
the strengths of adjacent triangle edges are equal.
\end{theorem}
\ifkdd\else
\begin{proof}
Consider an optimum $\bw$ for which this is not the case,
i.e., two adjacent triangle edges can be found that have different strength.
From Corollary~\ref{cor:permutations} we know that we can construct from this optimal solution
another optimal solution $\bw^=$ for which adjacent triangle edges do have the same strength,
equal to the average strength in $\bw$ of all triangle edges in the triangle clique they are part of.
Moreover, in $\bw^=$ all rays within the same bundle have the same strength,
equal to the average strength in $\bw$ of all rays in the bundle.
Let us denote the strength in $\bw^=$ of the $b$-th bundle to the triangle clique as $w^=_b$ (i.e., $b$ is an index to the bundle),
and the strength of the edges in the triangle clique as $w^=_c$.
We will prove that $\bw^=$ is not optimal, reaching a contradiction.

In particular, we will show that there exists a solution $\bw^*$ for which $w^*_b=w^=_b$ for all bundles $b$,
but for which the strength within the triangle clique is strictly larger: $w^*_c>w^=_c$.
We first note that the strengths of the triangle edges $w^*_c$ are bounded
in triangle constraints involving two rays and one triangle edge, 
namely $w^*_c \leq 2+(d-1)\cdot w^=_b$.
They are bounded also in triangle constraints involving only triangle edges,
namely $(2-d)\cdot w^*_c\leq 2$.
For $d\geq 2$ this constraint is trivially satisfied,
but not for $d<2$.
Thus, we know that 
$w^*_c = 2+\min_b \left\{(d-1)\cdot w^=_b\right\}$ for $d\ge 2$, and
$w^*_c = \min\left\{2+\min_b \left\{(d-1)\cdot w^=_b\right\},\frac{2}{2-d}\right\}$ for $d<2$.
If this optimal value for $w^*_c$ is larger than $w^=_c$
the contradiction is established.

First we show that $w^=_c < \frac{2}{2-d}$ when $d<2$, again by contradiction.
For each triangle $\{i,j,k\}$ in the triangle clique,
the following triangle inequality is the tightest:
$\max\{w_{ij},w_{ik},w_{jk}\} + \mbox{median}\{w_{ij},w_{ik},w_{jk}\} \leq 2+d\cdot\min\{w_{ij},w_{ik},w_{jk}\}$.
Averaging these constraints over all triangles within the triangle clique,
we obtain:
\begin{align}\label{eq:bla}
w^+ + w^0 &\leq 2 + d\cdot w^-
\end{align}
for some $w^+\geq w^0\geq w^-$ for which $w^=_c=\frac{1}{3}(w^++w^0+w^-)$.
Since we assumed (with the intention to reach a contradiction)
that not all $w_{ij}$ in the triangle clique are equal,
we also know that $\frac{w^++w^0}{2} > w^=_c > w^-$.
Given this, and if it were indeed the case that $w^=_c\geq \frac{2}{2-d}$,
Eq.~(\ref{eq:bla}) would imply that $2\cdot\frac{2}{2-d}<2+d\cdot w^-$,
and thus, $\frac{2}{2-d}<w^-$,
a contradiction.

Next, we show that $w^=_c < 2+\min_b \left\{(d-1)\cdot w^=_b\right\}$.
To show this, we need to distinguish two cases:
\begin{enumerate}
\item \emph{The bundle $b$ with smallest $(d-1)\cdot w^=_b$ has at least two different weights in $\bw$.}
Then, note that for each pair of ray strengths $w_{bi}$ and $w_{bj}$ from node $b$ to triangle-edge $(i,j)$,
the following bounds must hold: $w_{ij}\leq 2 + d\cdot \min\{x_{bi},x_{bj}\} - \max\{x_{bi},x_{bj}\}$.
Summing this over all $\{i,j\}$ and dividing by $\frac{n(n-1)}{2}$ where $n$ is the number of nodes in the triangle clique, yields:
$w^=_c\leq 2 + d\cdot w_b^- - w_b^+$ for some $w_b^-<w_b^+$ with $w_b=\frac{w_b^- + w_b^+}{2}$.
This means that $w^=_c < 2 + \min_b \left\{(d-1)\cdot x_b\right\}$, 
with a strict inequality since we assumed that there is at least one pair of rays $\{b,i\}$ and $\{b,j\}$ for which $w_{bi}<w_{bj}$.
Thus, this shows that $w^*_c>w^=_c$, and a contradiction is reached.
\item \emph{All rays in the bundle $b$ with smallest $(d-1)\cdot w^=_b$ have equal strength $w_b=w^=_b$ in $\bw$.}
In this case, we know that $w_{ij}\leq 2+(d-1)w^=_b$ (due to feasibility of the original optimum).
Again, averaging this over all triangle edges $\{i,j\}$, yields:
$w^=_c\leq 2+(d-1)w^=_b$, with equality only if all terms are equal (since the right hand side is independent of $i$ and $j$).
Thus, again a contradiction is reached.
\end{enumerate}
\end{proof}
\fi

Note that there do exist graphs for which not all optimal solutions have equal strengths within a bundle.
An example is shown in Fig.~\ref{fig:non-constant-bundle}.

\begin{figure}[t]
\begin{tikzpicture}

\tikzstyle{exedge} = [yafcolor5!80, thick, text=black!80]
\tikzstyle{exnode} = [thick, draw = yafcolor7!80, fill=white, circle, inner sep = 1pt, text=black, minimum width=11pt]

\node[exnode] (b1) at (0, 0)    {$b_1$};
\node[exnode] (b2) at (0, 3) {$b_2$};
\node[exnode] (x1) at (1, 1.5)  {$x_1$};
\node[exnode] (x2) at (2, 0.8)  {$x_2$};
\node[exnode] (x3) at (2, 2.2)  {$x_3$};
\node[exnode] (y) at (3.5, 1.5)  {$y$};
\node[exnode] (z1) at (5, 1.5) {$z_1$};
\node[exnode] (z2) at (5.8, 0) {$z_2$};
\node[exnode] (z3) at (5.8, 3) {$z_3$};
\node[exnode] (z4) at (6.8, 0) {$z_4$};
\node[exnode] (z5) at (6.8, 3) {$z_5$};
\node[exnode] (z6) at (7.6, 1.5) {$z_6$};

\draw[-, exedge, bend left = 10] (b1) to (x1);
\node (b1x1) at (0.3, 1.1) {$1/3$};
\draw[-, exedge, bend right = 10] (b1) to (x2);
\node (b1x2) at (1.05, 0.5) {$2/3$};
\draw[-, exedge, bend right = 10] (b1) to (x3);
\node (b1x3) at (0.8, 0.9) {$1/3$};

\draw[-, exedge, bend right = 10] (b2) to (x1);
\node (b2x1) at (0.3, 1.9) {$2/3$};
\draw[-, exedge, bend left = 10] (b2) to (x2);
\node (b2x2) at (1.1, 2.3) {$1/3$};
\draw[-, exedge, bend left = 10] (b2) to (x3);
\node (b2x3) at (1.6, 2.7) {$2/3$};

\draw[-, exedge, bend right = 5] (x1) to (x2);
\draw[-, exedge, bend right = 5] (x2) to (x3);
\draw[-, exedge, bend right = 5] (x3) to (x1);

\draw[-, exedge, bend left = 0] (x1) to (y);
\draw[-, exedge, bend right = 5] (x2) to (y);
\draw[-, exedge, bend left = 5] (x3) to (y);

\draw[-, exedge, bend left = 0] (y) to (z1);
\draw[-, exedge, bend right = 15] (y) to (z6);
\draw[-, exedge, bend left = 0] (y) to (z2);
\draw[-, exedge, bend left = 0] (y) to (z3);
\draw[-, exedge, bend left = 0] (y) to (z4);
\draw[-, exedge, bend left = 0] (y) to (z5);

\draw[-, exedge, bend left = 0] (z1) to (z2);
\draw[-, exedge, bend left = 0] (z1) to (z3);
\draw[-, exedge, bend left = 0] (z1) to (z4);
\draw[-, exedge, bend left = 0] (z1) to (z5);
\draw[-, exedge, bend left = 0] (z1) to (z6);
\draw[-, exedge, bend left = 0] (z2) to (z3);
\draw[-, exedge, bend left = 0] (z2) to (z4);
\draw[-, exedge, bend left = 0] (z2) to (z5);
\draw[-, exedge, bend left = 0] (z2) to (z6);
\draw[-, exedge, bend left = 0] (z3) to (z4);
\draw[-, exedge, bend left = 0] (z3) to (z5);
\draw[-, exedge, bend left = 0] (z3) to (z6);
\draw[-, exedge, bend left = 0] (z4) to (z5);
\draw[-, exedge, bend left = 0] (z4) to (z6);
\draw[-, exedge, bend left = 0] (z5) to (z6);

\end{tikzpicture}
\caption{\label{fig:non-constant-bundle} This graph is an example where an optimal solution of Problem~\ref{eq:LP2-STC} (with $d=2$) exists that is not constant within a bundle.
To see this, note that $y$ is the root of a bundle to both triangle cliques (the one with nodes $x_i$ and the one with nodes $z_i$).
Its rays to both bundles constrain each other in wedge constraints.
As the $z$ triangle clique is large, the optimal solution has the largest possible value for edges to those nodes.
This is achieved by assigning strengths of $1$ to $y$'s rays to $z_i$, and $0$ to $y$'s rays to $x_i$.
Then the triangle edges in the $z$ triangle clique can have strength $3$,
and the strengths between the $x$ nodes is $2$.
There are two other bundles to the $x$ triangle clique: from $b_1$ and $b_2$.
These constrain each other in wedges $(x_i,\{b_1,b_2\})$, such that edges from $b_1$ and $b_2$ to the same $x_i$ must sum to $1$ at the optimum.
Furthermore, triangles $\{b_i,x_j,x_k\}$ impose a constraint on the strength of those edges as:
$w_{b_ix_j} + w_{x_ix_k}\leq 2 + d\cdot x_{b_ix_k}$. For $d=2$ and $w_{x_jx_k}=2$, this gives:
$w_{b_ix_j}\leq 2\cdot x_{b_ix_k}$.
No other constraints apply. Thus, the (unequal) strengths for the edges in the bundles from $b_1$ and $b_2$ shown in the figure are feasible.
Moreover, this particular optimal solution is a vertex point of the feasible polytope%
\ifkdd\else
 (proof not given)
\fi. $1/2$ for each of those edges is also feasible.}
\end{figure}
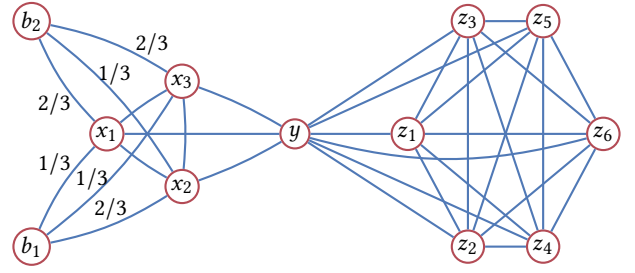

\subsection{An equivalent formulation for finding symmetric optima of Problem~\ref{eq:LP2-STC}}\label{sec:2var}

\oldnote[Tijl]{I plan to try rewriting this section by referring to a graph obtained
by edge-contracting all triangle edges.
Each triangle clique is a node in this contracted graph,
and each bundle is an edge to a clique node in this contracted graph.
This contracted graph does not contain any triangle edges,
i.e., each edge is part of a wedge.
As wedge constraints are tighter than triangle constraints for the non-slack case (the case we are concerned with),
this means that all we need is wedge constraints on this contracted graph,
and triangle constraints for edges between a triangle node and a normal one, or between two triangle nodes.
The objective function is then a weighted sum of all strengths,
with weight equal to the number of edges that were contracted to it (we could call it its `multiplicity').
That may be easier to explain, and it may be easier to write down the final optimization problem.}

Solutions that lack the symmetry properties 
specified in Corollary \ref{cor:permutations} essentially make arbitrary strength assignments.
Thus, it makes sense to constrain the search space to just those optimal solutions 
that exhibit these symmetries.%
\footnote{It would be desirable to search only for solutions that exhibit 
all symmetries guaranteed by Theorem~\ref{thm:automorphisms},
but given the algorithmic difficulty of enumerating all automorphisms, this is hard to achieve directly.
Also, realistic graphs probably contain few automorphisms other than the permutations within triangle cliques.
\ifkdd
\else
Section~\ref{sec:reducing-arbitrariness}
\fi
does however describe an indirect but still polynomial-time approach
for finding fully symmetric solutions.
}
In addition, exploiting symmetry leads to fewer variables, and thus, 
computational-efficiency gains.

In this section, we will refer to strength assignments 
that are invariant with respect to permutations 
within triangle cliques as \emph{symmetric}, for short.
The results here apply only to Problem~\ref{eq:LP2-STC}.

The set of free variables consists of one variable per triangle clique,
one variable per bundle,
and one variable per edge that is neither a triangle edge nor a ray in a bundle.
To reformulate Problem~\ref{eq:LP2-STC} in terms of this reduced set of variables,
it is convenient to introduce the \emph{contracted graph},
defined as the graph obtained by edge-contracting all triangle edges in $G$.
More formally:
\begin{definition}[Contracted graph]
Let $\sim$ denote the equivalence relation between nodes
defined as $i\!\sim\!j$ if and only if $i$ and $j$ are connected by a triangle edge.
Then, the contracted graph $\tG=(\tV,\tE)$ with 
$\tE\subseteq {\tV \choose 2}$
is defined as the graph for which $\tV=V/\sim$ (the quotient set of $\sim$ on $V$),
and for any $A,B\in\tV$,
it holds that $\{A,B\}\in\tE$ if and only if for all $i\in A$ and $j\in B$ 
it holds that $\{i,j\}\in E$.
\end{definition}

Figure~\ref{fig:contracted} illustrates these definitions for the graph from Fig.~\ref{fig:non-constant-bundle}.

\begin{figure}[t]
\begin{tikzpicture}

\tikzstyle{exedge} = [yafcolor5!80, thick, text=black!80]
\tikzstyle{exnode} = [thick, draw = yafcolor7!80, fill=white, circle, inner sep = 1pt, text=black, minimum width=11pt]

\node[exnode] (b1) at (0, 0)    {$\{b_1\}$};
\node[exnode] (b2) at (0, 2) {$\{b_2\}$};
\node[exnode] (x) at (1.5, 1)  {$\{x_i|i=1:3\}$};
\node[exnode] (y) at (3, 1)  {$\{y\}$};
\node[exnode] (z) at (4.5, 1) {$\{z_i|i=1:6\}$};

\draw[-, exedge, bend left = 0] (b1) to (x);

\draw[-, exedge, bend right = 0] (b2) to (x);

\draw[-, exedge, bend left = 0] (x) to (y);

\draw[-, exedge, bend left = 0] (y) to (z);

\end{tikzpicture}
\caption{\label{fig:contracted}
The contracted graph corresponding to the graph shown in Fig.~\ref{fig:non-constant-bundle}.
}
\end{figure}
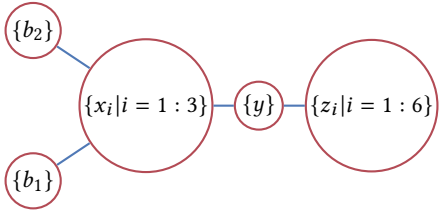


We now introduce a vector $\bw^t$ indexed by sets $A\subseteq V$, with $|A|\geq 2$,
with $w^t_A$ denoting the strength of the edges in the triangle clique $A\subseteq V$.
We also introduce a vector $\bw^w$ indexed by unordered pairs $\{A,B\}\in\tE$,
with $w^w_{AB}$ denoting the strength of the wedge edges between nodes in $A\subseteq V$ and $B\subseteq V$.
Note that if $|A|\geq 2$ or $|B|\geq 2$, these edges are rays in a bundle.

With this notation,
we can state the symmetrized problem as:
\begin{align}\tag{LP2sym}\label{eq:LP2-STC-sym}
\max_{\bw^t,\bw^w} & \sum_{A\in \tV:|A|\geq 2} \frac{|A|(|A|-1)}{2}w^t_{A}
 + \sum_{\{A,B\}\in \tE} |A||B|w^w_{AB},
\end{align}
\begin{align}
\mbox{s.t.}\quad & w^w_{AB} + w^w_{AC} \leq 1, &\mbox{for all }& (A,\{B,C\})\in\tcW,\label{eq:wedge-sym}\\
& w^t_{A} \leq 2 + (d-1)\cdot w^w_{AB}, &\mbox{for all }& \{A,B\}\in\tE,|A|\geq 2,\label{eq:triangle-sym-2var}\\
& w^t_{A}\leq \frac{2}{2-d} \mbox{\ \ (if } d<1\mbox{)} , &\mbox{for all }& A\in\tV, |A|\geq 3,\label{eq:triangle-sym-1var}\\
& w^t_{A}\geq 0, &\mbox{for all }& A\in \tV,|A|\geq 2,\label{eq:pos-sym-1}\\
& w^w_{AB}\geq 0, &\mbox{for all }& \{A,B\}\in \tE.\label{eq:pos-sym-2}
\end{align}
The following theorem shows that there is a one-to-one mapping between optimal solutions to this problem
and \emph{symmetric} optimal solutions to Problem~\ref{eq:LP2-STC},
setting $w_{ij}=w^t_A$ if and only if $i,j\in A$,
and $w_{ij}=w^w_{AB}$ if and only if $i\in A,j\in B$.

\oldnote[Tijl]{For now, this section applies only to the non-slack variants,
although I think very similar results hold for the slack variants.
However I haven't found similarly short proofs for the slack variants yet.}

\oldnote[Tijl]{Rather than relaxing the problem further, here we are going to \emph{add} more constraints: equality constraints for all pairs of isomorphic edges. Here we will show how doing this allows us to simplify the structure of the problem for subsequent analysis and algorithm development.}

\begin{theorem}[Problem~\ref{eq:LP2-STC-sym} finds symmetric solutions of Problem~\ref{eq:LP2-STC}]\label{th:LP2-STC-sym}
The set of optimal symmetric optimal solutions of Problem~\ref{eq:LP2-STC}
is equivalent to the set of all optimal solutions of Problem~\ref{eq:LP2-STC-sym}.
\end{theorem}
\ifkdd\else
\begin{proof}
It is easy to see that for a symmetric solution,
the objective functions of Problems~\ref{eq:LP2-STC-sym} and \ref{eq:LP2-STC} are identical.

Thus, it suffices to show that:
\begin{enumerate}
\item the feasible region of Problem~\ref{eq:LP2-STC-sym}
is contained within the feasible region of Problem~\ref{eq:LP2-STC},
\item the set of \emph{symmetric} feasible solutions of Problem~\ref{eq:LP2-STC} is contained within the feasible region of Problem~\ref{eq:LP2-STC-sym}.
\end{enumerate}
The latter is immediate,
as all constraints in Problem~\ref{eq:LP2-STC-sym} are directly derived from those in Problem~\ref{eq:LP2-STC}
(see rest of the proof for clarification),
apart from the reduction in variables which does nothing else than imposing symmetry.

To show the former,
we need to show that all constraints of Problem~\ref{eq:LP2-STC} are satisfied. 
This is trivial for the positivity constraints (\ref{eq:pos-sym-1}) and (\ref{eq:pos-sym-2}).
The wedge inequalities (\ref{eq:wedge-sym}) are also accounted in Problem~\ref{eq:LP2-STC-sym},
and thus trivially satisfied, too.

We consider three types of triangle constraints in Problem~\ref{eq:LP2-STC}:
those involving two rays from the same bundle and one triangle edge with the triangle edge strength on the left hand side of the $<$ sign,
those involving two rays from the same bundle and one triangle edge with the triangle edge strength on the right hand side of the $>$ sign,
and those involving three triangle edges.

Constraint~(\ref{eq:triangle-sym-2var}) covers all triangle constraints involving
two rays (from the same bundle) and one triangle edge,
with the triangle edge strength being upper bounded.
Indeed, $w^t_A$ is the strength of the triangle edges between nodes in $A$,
and $w^w_{AB}$ is the strength of the edges in the bundle from any node in $B$
to the two nodes connected by any triangle edge in $A$.

Triangle constraints involving two rays and one triangle edge
that lower bound the triangle edge strength are redundant and can thus be omitted.
Indeed, they can be stated as $w^w_{AB}+w^w_{AB}\leq 2+d\cdot w^t_{A}$,
which is trivially satisfied as each wedge edge has strength at most~$1$.

Finally, triangle constraints involving three triangle edges within $A$ reduce to $w^t_A+w^t_A\leq 2+ d\cdot w^t_A$.
For $d\geq 2$ this constraint is trivially satisfied.
For $d<2$ it reduces to $w^t_A\leq \frac{2}{2-d}$.
For $2>d\geq 1$,
this constraint is also redundant with the triangle constraints involving the triangle edge and two rays,
which imply an upper bound of at most $d+1\leq \frac{2}{2-d}$ for $d\geq 1$ (namely for the ray strengths equal to $1$).
Thus, constraint $w^t_A\leq \frac{2}{2-d}$ must be included in Problem~\ref{eq:LP2-STC-sym} only for $d<1$.
Finally, note that such triangle constraints are only possible in triangle cliques $A$ with $|A|\geq 3$.
\end{proof}
\fi

\subsection{The vertex points of the feasible polytope of Problem~\ref{eq:LP2-STC}}

The following theorem generalizes the well-known half-integrality result for Problem~\ref{eq:LP1-STC} \cite{NeT:75}
to Problem~\ref{eq:LP2-STC-sym}.

%

\oldnote[Tijl]{First recap the known result for the simplest relaxation, here in the preamble of this section.}

\oldnote[Tijl]{Note: results here do not require $d$ to be integer.}


\ifkdd\else
\begin{theorem}[Vertices of the feasible polytope]
\label{thm:verticespoly}
On the vertex points of the feasible polytope of Problem~\ref{eq:LP2-STC-sym}, 
the strengths of the wedge edges take values $w^w_{AB}\in\left\{0,\frac{1}{2},1\right\}$, 
and the strengths of the triangle edges take values $w^t_{A}\in\left\{0,2,\frac{d+3}{2},d+1\right\}$ for $d\geq 1$, 
or $w^t_{A}\in\left\{0,\frac{2}{2-d},d+1,\frac{d+3}{2},2\right\}$ for $d<1$.
\end{theorem}
\begin{proof}
Assume the contrary, i.e., 
that a vertex point of this convex feasible polytope 
can be found that has a different value for one of the edge strengths.
To reach a contradiction, we will nudge the \emph{wedge edges'} strengths $w^w_{AB}$ as follows:
\begin{itemize}
\item[--] if $w^w_{AB}\in(0,\frac{1}{2})$, add $\epsilon$,
\item[--] if $w^w_{AB}\in(\frac{1}{2},1)$, subtract $\epsilon$.
\end{itemize}
Note that $0\leq w^w_{AB}\leq 1$ for all wedge edges (due to the wedge constraints in Eq.~(\ref{eq:wedge-sym})).
Thus, all wedge edge strengths that are not exactly equal to $0$, $\frac{1}{2}$, or $1$ will be nudged.

For the \emph{triangle edges} we need to distinguish between $d\geq 1$ and $d<1$.
For $d\geq 1$, we nudge their strengths as follows:
\begin{itemize}
\item[--] if $w^t_A\in(0,2)$, add $\epsilon$,
\item[--] if $w^t_{A}\in\left(2,2+(d-1)\cdot\frac{1}{2}\right)=\left(2,\frac{d+3}{2}\right)$, add $(d-1)\epsilon$,
\item[--] if $w^t_{A}\in\left(2+(d-1)\cdot\frac{1}{2}, 2+(d-1)\cdot 1\right)=\left(\frac{d+3}{2}, d+1\right)$, subtract $(d-1)\epsilon$.
\end{itemize}
For $d<1$, we nudge the strenghts as follows:
\begin{itemize}
\item[--] if $w^t_{A}\in\left(0,\frac{2}{2-d}\right)\cup\left(\frac{2}{2-d},d+1\right)$, add $\epsilon$,
\item[--] if $w^t_{A}\in\left(2+(d-1)\cdot 1, 2+(d-1)\cdot\frac{1}{2}\right)=\left(d+1, \frac{d+3}{2}\right)$, subtract $(d-1)\epsilon$.
\item[--] if $w^t_{A}\in\left(2+(d-1)\cdot\frac{1}{2},2\right)=\left(\frac{d+3}{2},2\right)$, add $(d-1)\epsilon$.
\end{itemize}
Note that for $d\geq 1$, it holds that $0\leq w^t_{A}\leq d+1$ for all triangle edges (due to the triangle constraints in Eq.~(\ref{eq:triangle-sym-2var}) with ray strength equal to $1$).
For $d<1$, it holds that $0\leq w^t_{A}\leq 2$ for all triangle edges (due to the triangle constraints in Eq.~(\ref{eq:triangle-sym-2var}) with ray strength equal to $0$).
Thus, also all triangle edge strengths that are not of one of the values specified in the theorem statement will be nudged.

For sufficiently small $|\epsilon|$ no loose constraint will become invalid by this.
Furthermore, it is easy to verify that strengths in tight constraints are nudged in corresponding directions,
such that all tight constraints remain tight and thus valid.
Now, this nudging can be done for positive and negative $\epsilon$,
yielding two new feasible solutions of which the average is the supposed vertex point of the polytope---a contradiction.
\end{proof}
\fi

\ifkdd
\begin{theorem}[Vertices of the optimal face of the feasible polytope]\label{thm:half-integrality}
\else
\begin{corollary}\label{cor:half-integrality}
\fi
On the vertices of the optimal face of the feasible polytope of Problem~\ref{eq:LP2-STC-sym},
the strengths of the wedge edges take values $w^w_{AB}\in\left\{0,\frac{1}{2},1\right\}$,
and the strengths of the triangle edge take values $w^t_{A}\in\left\{2,\frac{d+3}{2},d+1\right\}$ if $d\geq 1$,
or $w^t_{A}\in\left\{\frac{2}{2-d},d+1,\frac{d+3}{2},2\right\}$ if $d<1$.
Moreover, for $d<1$, triangle edge strengths for $|A|\geq 3$ are all equal to $w^t_A=\frac{2}{2-d}$ throughout the optimal face of the feasible polytope.
\ifkdd
\end{theorem}
\else
\end{corollary}
\fi
\ifkdd\else
\begin{proof}
A strength of $0$ for a triangle edge can never be optimal, as triangle edges are upper bounded by at least $2$ for $d\geq 1$,
and the objective function is an increasing function of the edge strengths.
The second statement follows from the fact that $\frac{2}{2-d}$ is the smallest possible value for triangle edges when $d<1$,
and Eq.~(\ref{eq:triangle-sym-1var}) bounds the triangle edge strengths in triangle cliques $A$ with $|A|\geq 3$ to that value.
Thus, it is the only possible value for the vertex points of the optimal face of the feasible polytope,
and thus for that entire optimal face.
\end{proof}
\fi
This means that there always exists an optimal solution to Problem~\ref{eq:LP2-STC-sym}
where the edge strengths belong to these small sets of possible values.
Note that the symmetric optima of Problem~\ref{eq:LP2-STC} coincide with those of Problem~\ref{eq:LP2-STC-sym},
such that this result obviously also applies to the symmetric optima of~\ref{eq:LP2-STC}.

\section{Algorithms}\label{sec:algorithms}

In this section we discuss algorithms for
solving the edge-strength inference problems
\ref{eq:LP1-STC}, \ref{eq:LP2-STC}, \ref{eq:LP3-STC}, and \ref{eq:LP4-STC}.
\ifkdd\else
The final subsection~\ref{sec:reducing-arbitrariness} also discusses a number of ways to further reduce the arbitrariness of the optimal solutions.
\fi

\subsection{Using generic LP solvers}

First, all proposed formulations are linear programs (LP),
and thus, standard LP solvers can be used. 
In our experimental evaluation we used CVX~\cite{cvx} from within Matlab,
and MOSEK \cite{mosek} as the solver that implements an interior-point method.

Interior-point algorithms for LP run in polynomial time,
namely in $\bigO(n^{3}L)$ operations, 
where $n$ is the number of variables, and 
$L$ is the number of digits in the problem specification~\cite{mehrotra1993finding}.
For our problem formulations, 
$L$ is proportional to the number of constraints.
In particular, 
problem \ref{eq:LP1-STC} has $|E|$ variables and $|\cW|$ constraints, 
problem \ref{eq:LP2-STC} has $|E|$ variables and $|\cW|+|\cT|$ constraints, 
and 
problems \ref{eq:LP3-STC} and \ref{eq:LP4-STC} have $|E|+|\bE|$ variables and $|\cW|+|\cT|$ constraints.
\ifkdd\else
Here $|E|$ is the number of edges in the input graph, 
$|\cW|$ the number of wedges, and
$|\cT|$ the number of triangles.
\fi

Today, the development of primal-dual methods and practical improvements
ensure convergence that is often much faster than this worst-case complexity.
\oldnote[Tijl$\rightarrow$Aris]{Can you say something more about this?}
\oldnote[Aris]{I do not really know, and I was not able to find something. 
Also I do not think that this is something that we should discuss in more detail.}
Alternatively, one can use the Simplex algorithm, 
which has worst-case exponential running time, 
but is known to yield excellent performance in practice \cite{spielman2004smoothed}.

\oldnote[Tijl]{General LP solvers, some notes on properties and complexity. Also here we could implicitly impose symmetry in the solution by using only one variable per triangle clique and one variable per bundle.}

\oldnote[Tijl]{Of course, a generic LP solver can also be used to find a solution of the 2VAR formulation for integer $d\geq 1$.}

\subsection{Using the Hochbaum-Naor algorithm}

For rational $d$,
we can exploit the special structure of Problems \ref{eq:LP1-STC} and \ref{eq:LP2-STC}
and solve them using more efficient combinatorial algorithms.
In particular, 
the algorithm of Hochbaum and Naor~\cite{hochbaum1994simple}
is designed for a family of integer problems named 2VAR problems.
2VAR problems are integer programs (IP) with 2 variables per constraint 
of the form $a_{k}x_{i_k} - b_{k}x_{j_k} \geq c_{k}$ with rational $a_k, b_k,$ and $c_k$,
in addition to integer lower and upper bounds on the variables.
A 2VAR problem is called \emph{monotone} if the coefficients $a_{k}$ and $b_{k}$ have the same sign. Otherwise the IP is called \emph{non-monotone}.
The algorithm of Hochbaum and Naor~\cite{hochbaum1994simple}
gives an \emph{optimal integral} solution for monotone IPs
and an \emph{optimal half-integral} solution for non-monotone IPs.
The running time of the algorithm is pseudopolynomial: polynomial in the range (difference between lower bound and upper bound) of the variables. 
More formally, it is $\bigO(n\Delta^2(n+r))$, 
where $n$ is the number of variables, 
$r$ is the number of constraints, and $\Delta$ is maximum range size. 
\ifkdd
\else
For completeness, we briefly discuss the problem and algorithm for solving it below.

\paragraph{The monotone case.}
We first consider an IP with monotone inequalities: 
\begin{align}\tag{monotone IP}\label{eq:monotonesystem}
\max\, & \sum_{i=1}^{n} d_ix_i,&\\
\st\, & a_{k}x_{i_k} - b_{k}x_{j_k} \geq c_{k}      &\mbox{ for } k = 1,\ldots,r,\\
& \ell_i\leq x_i\leq u_i,\quad\quad x_i\in \mathbb{Z}, &\mbox{ for } i = 1,\ldots,n,
\end{align}
where $a_k$, $b_k$, $c_k$, 
and $d_i$
are rational, 
while 
$\ell_i$ and $u_i$
are integral. 
The coefficients $a_k$ and $b_k$ 
have the same sign, 
and $d_i$ 
can be negative.

The algorithm is based on constructing a weighted directed 
graph $G'=(V',E')$ and finding a minimum $s\mbox{-}t$ cut on~$G'$. 

For the construction of the graph $G'$, 
for each variable $x_i$ in the IP we create a set of $(u_i-\ell_i+1)$ nodes $\{v_{ip}\}$, 
one for each integer~$p$ in the range $[\ell_i, u_i]$. 
An auxiliary source node $s$ and a sink node~$t$ are added. 
All nodes that correspond to positive integers are denoted by $V^{+}$, and 
all nodes that correspond to non-positive integers are denoted by $V^{-}$.

The edges of $G'$ are created as follows:
First, we connect the source $s$ to all nodes $v_{ip}\in V^{+}$, 
with $\ell_i+1 \le p \le u_i$. 
We also connect all nodes $v_{ip}\in V^{-}$, with $\ell_i+1 \le p \le u_i$,
to the sink node $t$.
All these edges have weight $|d_i|$.
The rest of the edges described below have infinite weight.

For the rest of the graph, 
we add edges from $s$ to all nodes $v_{ip}$ with $p=\ell_i$---both in $V^+$ and $V^-$. 
For all $\ell_i+1 \le p \le u_i$, 
the node $v_{ip}$ is connected to $v_{i(p-1)}$ by a directed edge. 
Let $q_k(p)=\lceil {\frac{c_{k}+ b_{k}p}{a_k}} \rceil$. 
For each inequality $k$ we connect node $v_{j_kp}$, 
corresponding to $x_{j_k}$ with $\ell_{j_k} \le p \le u_{j_k}$, 
to the node $v_{i_kq}$, 
corresponding to $x_{i_k}$ where $q=q_k(p)$. 
If $q_k(p)$ is below the feasible range $[\ell_{i_k}, u_{i_k}]$, 
then the edge is not needed. 
If $q_k(p)$ is above this range, then node $v_{j_kp}$ must be connected to the sink $t$.

Hochbaum and Naor~\cite{hochbaum1994simple} show 
that the optimal solution of (\ref{eq:monotonesystem}) 
can be derived from the source set $S$ of minimum $s$-$t$ 
cuts on the graph $G'$ by setting $x_i=\max \{p\mid v_{ip}\in S\}$. 
The complexity of this algorithm is dominated by solving the minimum $s$-$t$ cut problem, 
which is $\bigO(|V'||E'|)= \bigO(n\Delta^2(n+r))$, 
where $n$ is the number of variables in the (\ref{eq:monotonesystem}) problem, 
$r$ is the number of constraints, 
and $\Delta$ is maximum range size $\Delta=max_{i=[1,n]} (u_i-\ell_i+1)$. 
Recall that in Problem \ref{eq:LP1-STC} 
the number of variables is equal to the number of edges in the original graph $m$, 
the number of constraints is equal to the number of wedges $|\cW|$, 
while the range size $\Delta$ is constant.

Note also that despite the relatively high worst-case complexity, 
in practice the graph is sparse and finding the cut is fast.

\paragraph{Monotonization and half-integrality}
A non-monotone IP with two variables per constraint is NP-hard. 
Edelsbrunner et al.~\cite{edelsbrunner1989testing} 
showed that a non-monotone IP with two variables per constraint has half-integral solutions, 
which can be obtained by the following monotonization procedure.
Consider a non-monotone IP:

\begin{align}\tag{non-monotone IP}\label{eq:nonmonotone}
\max\, & \sum_{i=1}^{n} d_ix_i,&\\
\st\, & a_{k}x_{i_k} + b_{k}x_{j_k} \geq c_{k}& \mbox{ for all } k = 1,\ldots,m,\\
& \ell_i\leq x_i\leq u_i,\quad\quad x_i\in \mathbb{Z}, &\mbox{ for all } i = 1,\ldots,n,
\end{align}
with no constraints on the signs of $a_{k}$ and $b_{k}$. 

For monotonization we replace each variable 
$x_i$ by $x_i=\frac{x_i^+ - x_i^-}{2}$, 
where $\ell_i\leq x_i^+\leq u_i$ and $-u_i\leq x_i^-\leq -\ell_i$.
Each non-monotone inequality 
($a_k$ and $b_k$ having the same sign) 
$a_{k}x_{i_k} + b_{k}x_{j_k} \geq c_{k}$ is replaced by a pair:
\begin{align}
a_{k}x_{i_k}^+ - b_{k}x_{i_k}^- \geq c_{k}\\
-a_{k}x_{i_k}^- + b_{k}x_{i_k}^+ \geq c_{k}
\end{align}

Each monotone inequality $\bar a_{k}x_{i_k} - \bar b_{k}x_{j_k} \geq c_{k}$ is replaced by:
\begin{align}
\bar a_{k}x_{i_k}^+ - \bar b_{k}x_{j_k}^+ \geq c_{k}\\
-\bar a_{k}x_{i_k}^- + \bar b_{k}x_{j_k}^- \geq c_{k}
\end{align}

The objective function is replaced by $\sum_{i=1}^{n}{ \frac{1}{2} d_ix^+_i -\frac{1}{2} d_ix^-_i}$.

By construction, the resulting \ref{eq:monotonesystem} is a half-integral relaxation of 
the Problem (\ref{eq:nonmonotone}).
\fi

\ref{eq:LP1-STC} is a (non-monotone) 2VAR system, 
so that it can directly be solved by the algorithm of Hochbaum and Naor.
Problem~\ref{eq:LP2-STC}, however, is not a 2VAR problem,
such that the Hochbaum and Naor algorithm is not directly applicable.
Yet for integer $d\geq 1$, Problem~\ref{eq:LP2-STC-sym} \emph{is} a 2VAR problem.
The lower bound on each of the variables is $0$,
and the upper bound is equal to $1$ for the wedge edges and $d+1$ for the triangle edges---i.e., both lower and upper bound are integers.
For rational $d$,
the upper bound is $\max\{2,d+1\}$, which may be rational,
but reformulating the problem in terms of $a\cdot w_{ij}$ for $a$ the smallest integer for which $a\cdot d$ is integer turns it into a 2VAR problem again.
Thus, for $d$ rational, finding one of the \emph{symmetric} solutions of Problem~\ref{eq:LP2-STC} \emph{can} be done using Hochbaum and Naor's algorithm.
Moreover, this symmetric solution will immediately be one of the half-integral solutions we know exist from \ifkdd\else Corollary~\ref{cor:half-integrality}\fi.

%
%
%


\note[Tijl$\rightarrow$Aris, Polina]{What is the computational complexity,
for comparison with interior point methods?}

\note[Polina]{Interior point methods seems to be an active field, people study lots of special cases and usually consider $m=n$ in complexity analysis which is confusing. 
The best, that we could find now has complexity of $O(n^3/\log(n))$. 
The classic result is $O(m^{3/2} n^2 L) = O( m^{3/2} n^2 \log(m+n))$ where $L$ is the number of bits, needed to represent the problem.
N. Karmarkar. A new polynomial-time algorithm for linear programming. Combinatorica,1984.
We can either use it or let me know and I'll try to figure out what is the most recent development.}

\ifkdd\else
\subsection{Approaches for further reducing arbitrariness}\label{sec:reducing-arbitrariness}

\note[Aris]{This section is a bit complex. 
There are a few different ideas running in parallel. 
Also, it is not entirely clear to me, so as to improve it.
I suggest to simplify it. 
One suggestion is to start with a single approach, explain it more clearly, 
and then only at the end mention other approaches that can be solved in a similar way. }

As pointed out in Sec.~\ref{sec:2var},
Problem~\ref{eq:LP2-STC-sym} does not impose symmetry with respect to \emph{all} graph automorphisms,
as it would be impractical to enumerate them.
However, in Sec.~\ref{sec:full-symmetry} below we discuss an efficient (polynomial-time) algorithm
that is able to find a solution that satisfies all such symmetries,
without the need to explicitly enumerate all graph automorphisms.

Furthermore, in Sec.~\ref{sec:optimalface} we discuss a strategy for reducing arbitrariness
that is not based on finding a fully symmetric solution.
This alternative strategy is to characterize the entire optimal face of the feasible polytope,
rather than selecting a single optimal (symmetric) solution from it.
We furthermore propose a number of different algorithms implementing this strategy,
which run in polynomial time as well.


Several algorithms discussed below exploit the following characterization of the optimal face.
As an example, 
and with $o^*$ the value of the objective at the optimum,
for Problem~\ref{eq:LP2-STC-sym} this characterization is:
\begin{align*}
\mathcal{P}^*=\Big\{\bw \mid & \sum_{\{i,j\}\in E}w_{ij}=o^*,&&\\
& w_{ij} + w_{ik} \leq 1, &\mbox{for all }& (i,\{j,k\})\in\cW,\\
& w_{ij}+w_{ik} \leq 2 + d\cdot w_{jk}, &\mbox{for all }& \{i,j,k\}\in\cT,\\
& w_{ij}\geq 0, &\mbox{for all }& \{i,j\}\in E.\Big\}
\end{align*}
It is trivial to extend this to the optimal faces of the other problems.

\subsubsection{Invariance with respect to all graph automorphisms}\label{sec:full-symmetry}

Here we discuss an efficient algorithm to find a fully symmetric solution,
without explicitly having to enumerate all graph automorphisms.

Given the optimal value of the objective function of (for example) Problem~\ref{eq:LP2-STC},
consider the following problem which finds a point in the optimal face of the feasible polytope
that minimizes the sum of squares of all edge strengths:
\begin{align}\tag{LP2fullsym}\label{eq:LP2-STC-full-sym}
\min_{\bw} & \sum_{\{i,j\}\in E} (w_{ij})^2,&&\\
\st &\bw\in\mathcal{P}^*.\nonumber
\end{align}
As $\mathcal{P}^*$ is a polytope,
this is a Linearly Constrained Quadratic Program (LCQP),
which again can be solved efficiently using interior point methods.

\begin{theorem}[Problem~\ref{eq:LP2-STC-full-sym} finds a solution symmetric with respect to all graph automorphisms]
The edge strength assignments that minimize Problem~\ref{eq:LP2-STC-full-sym}
are an optimal solution to Problem~\ref{eq:LP2-STC}
that is symmetric with respect to all graph automorphisms.
\end{theorem}
\begin{proof}
Let us denote the optimal vector of weights found by solving
Problem~\ref{eq:LP2-STC-full-sym} as $\bw^*$.
It is clear that $\bw^*$ is an optimal solution to Problem~\ref{eq:LP2-STC},
as it is constrained to be such.

Now, we will prove symmetry by contradiction:
let us assume there is a graph automorphism $\alpha\in\cA$ with respect to which it is not symmetric,
such that there exists a set of edges $\{i,j\}\in E$ for which
$w^*_{ij}\neq w^*_{\alpha(i)\alpha(j)}$.
Due to convexity, $\bw^{**}$ with $w^{**}_{ij}=\frac{w^*_{ij}+w^*_{\alpha(i)\alpha(j)}}{2}$ is then also
a solution to Problem~\ref{eq:LP2-STC-full-sym} and thus to Problem~\ref{eq:LP2-STC}.
However, since $a^2+b^2\geq 2\left(\frac{a+b}{2}\right)^2$ for any $a,b\in\mathbf{R}$,
$\bw^{**}$ has a smaller value for the objective of Problem~\ref{eq:LP2-STC-full-sym},
such that $\bw^*$ cannot be optimal---a contradiction.
\end{proof}
For simplicity of notation,
we explained this strategy for Problem~\ref{eq:LP2-STC},
but of course it is computationally more attractive to seek a solution within
the optimal face of the feasible polytope for Problem~\ref{eq:LP2-STC-sym}.

\subsubsection{Characterizing the entire optimal face of the feasible polytope}\label{sec:optimalface}

Here, we discuss an alternative strategy for reducing arbitrariness,
which is to characterize the entire optimal face of the feasible polytope of the proposed problem formulations,
rather than to select a single (possibly arbitrary) optimal solution from it.
Specifically, we propose three algorithmic implementations of this strategy.

The first algorithmic implementation of this strategy exactly characterizes the range of the strength of each edge amongst the optimal solutions.
This range can be found by solving, for edge strength $w_{ij}$ for each $\{i,j\}\in E$,
two optimization problems:
\begin{align*}
\max_{\bw}\, &w_{ij},\\
\st&\bw\in\mathcal{P}^*.
\end{align*}
and
\begin{align*}
\min_{\bw}\, &w_{ij},\\
\st&\bw\in\mathcal{P}^*.
\end{align*}
This is again an LP, and thus requires polynomial time.
Yet, it is clear that this approach is impractical,
as the number of such optimization problems to be solved is twice the number of variables in the original problem.

\oldnote[Tijl]{By solving a variant of the linear program, one for each edge}

The second algorithmic implementation of this strategy is computationally much more attractive,
but quantifies the range of each edge strength only partially.
It exploits the fact that the strengths at the vertex points of the optimal face belong to a finite set of values.
Thus, given any optimal solution,
we can be sure that for each edge,
there exists an optimal solution for which any given edge's strength is equal to the smallest value within that set equal to or exceeding the value in that optimal solution,
as well as a for which it is equal to the largest value within that set equal to or smaller than the value in that optimal solution.
To ensure this range is as large as possible,
it is beneficial to avoid finding vertex points of the feasible polytope,
and more generally points that do not lie within the \emph{relative interior} of the optimal face.
This can be done in the same polynomial time complexity as solving the LP itself,
namely $\bigO(n^3L)$ where $L$ is the input length of the LP \cite{mehrotra1993finding}.
This could be repeated several times with different random restarts to yield wider intervals for each edge strength.

\oldnote[Tijl]{Using interior point method for finding point in relative interior of optimal face}

The third implementation is to uniformly sample points (i.e., optimal solutions) from the optimal face $\mathcal{P}^*s$.
A recent paper \cite{chen2017fast} details an MCMC algorithm with polynomial mixing time for achieving this.

\oldnote[Tijl]{Also mention the possibility of sampling from the optimal face, as further work.}

\fi

\section{Empirical results}\label{sec:empirical}

This section contains the main empirical findings.
The code used in the experiments is available at \url{https://bitbucket.org/ghentdatascience/stc-code-public}.

\subsection{Qualitative analysis}
\begin{figure}[t]
\begin{subfigure}[b]{0.4\textwidth}
\centering
\resizebox{\linewidth}{!}{
\begin{tikzpicture}
\tikzstyle{exedge} = [yafcolor5!80, thick, text=black!80]
\tikzstyle{exedge2} = [yafcolor2!80, very thick, text=black!80]
\tikzstyle{exnode} = [thick, draw = yafcolor7!80, fill=white, circle, inner sep = 1pt, text=black, minimum width=11pt]

\node[exnode] (v1) at (0, 3)  {$1$};
\node[exnode] (v3) at (0, 1)  {$3$};
\node[exnode] (v2) at (0, 2)  {$2$};
\node[exnode] (v4) at (2, 2)  {$4$};

\node[exnode] (v5) at (4, 2)  {$5$};
\node[exnode] (v6) at (5.2, 2)  {$6$};
\node[exnode] (v7) at (6, 3)  {$7$};
\node[exnode] (v8) at (6, 1)  {$8$};

\draw[-, exedge, bend left = 0] (v1) to (v2);
\node (v1v2) at (-0.3,2.5) {$0$};

\draw[-, exedge, bend left = 0] (v3) to (v2);
\node (v3v2) at (-0.3,1.5) {$1$};

\draw[-, exedge, bend left = 0] (v2) to (v4);
\node (v4v2) at (1,2.15) {$1$};

\draw[-, exedge, bend left = 0] (v1) to (v4);
\node (v1v4) at (1.15,2.65) {$0$};

\draw[-, exedge, bend left = 0] (v3) to (v4);
\node (v3v4) at (1.15,1.35) {$1$};

\draw[-, exedge, bend left = 0] (v4) to (v5);
\node (v4v5) at (3,2.15) {$0$};

\draw[-, exedge, bend left = 0] (v5) to (v6);
\node (v5v6) at (4.7,2.15) {$1$};

\draw[-, exedge, bend left = 0] (v5) to (v7);
\node (v5v7) at (4.9,2.65) {$1$};

\draw[-, exedge, bend left = 0] (v5) to (v8);
\node (v5v8) at (4.9,1.35) {$1$};

\draw[-, exedge2, bend left = 0] (v6) to (v7);
\node (v6v7) at (5.7,2.45) {$1$};

\draw[-, exedge2, bend left = 0] (v6) to (v8);
\node (v6v8) at (5.7,1.55) {$1$};

\draw[-, exedge2, bend left = 0] (v7) to (v8);
\node (v7v8) at (6.15,2) {$1$};
\end{tikzpicture}}
\caption{\STCbinary\ Greedy approximation.}
\label{fig:toybin}
\end{subfigure}

\begin{subfigure}[b]{0.4\textwidth}
\centering
\resizebox{\linewidth}{!}{
\begin{tikzpicture}
\tikzstyle{exedge} = [yafcolor5!80, thick, text=black!80]
\tikzstyle{exedge2} = [yafcolor2!80, very thick, text=black!80]
\tikzstyle{exnode} = [thick, draw = yafcolor7!80, fill=white, circle, inner sep = 1pt, text=black, minimum width=11pt]

\node[exnode] (v1) at (0, 3)  {$1$};
\node[exnode] (v3) at (0, 1)  {$3$};
\node[exnode] (v2) at (0, 2)  {$2$};
\node[exnode] (v4) at (2, 2)  {$4$};

\node[exnode] (v5) at (4, 2)  {$5$};
\node[exnode] (v6) at (5.2, 2)  {$6$};
\node[exnode] (v7) at (6, 3)  {$7$};
\node[exnode] (v8) at (6, 1)  {$8$};

\draw[-, exedge, bend left = 0] (v1) to (v2);
\node (v1v2) at (-0.3,2.5) {$1/2$};

\draw[-, exedge, bend left = 0] (v3) to (v2);
\node (v3v2) at (-0.3,1.5) {$1/2$};

\draw[-, exedge, bend left = 0] (v2) to (v4);
\node (v4v2) at (1,2.15) {$1$};

\draw[-, exedge, bend left = 0] (v1) to (v4);
\node (v1v4) at (1.15,2.65) {$1/2$};

\draw[-, exedge, bend left = 0] (v3) to (v4);
\node (v3v4) at (1.15,1.35) {$1/2$};

\draw[-, exedge, bend left = 0] (v4) to (v5);
\node (v4v5) at (3,2.15) {$0$};

\draw[-, exedge, bend left = 0] (v5) to (v6);
\node (v5v6) at (4.7,2.15) {$1$};

\draw[-, exedge, bend left = 0] (v5) to (v7);
\node (v5v7) at (4.9,2.65) {$1$};

\draw[-, exedge, bend left = 0] (v5) to (v8);
\node (v5v8) at (4.9,1.35) {$1$};

\draw[-, exedge2, bend left = 0] (v6) to (v7);
\node (v6v7) at (5.7,2.45) {$1$};

\draw[-, exedge2, bend left = 0] (v6) to (v8);
\node (v6v8) at (5.7,1.55) {$1$};

\draw[-, exedge2, bend left = 0] (v7) to (v8);
\node (v7v8) at (6.15,2) {$1$};
\end{tikzpicture}}
\caption{\ref{eq:LP1-STC}.}
\label{fig:toylp1}
\end{subfigure}

\begin{subfigure}[b]{0.4\textwidth}
\centering
\resizebox{\linewidth}{!}{
\begin{tikzpicture}
\tikzstyle{exedge} = [yafcolor5!80, thick, text=black!80]
\tikzstyle{exedge2} = [yafcolor2!80, very thick, text=black!80]
\tikzstyle{exnode} = [thick, draw = yafcolor7!80, fill=white, circle, inner sep = 1pt, text=black, minimum width=11pt]

\node[exnode] (v1) at (0, 3)  {$1$};
\node[exnode] (v3) at (0, 1)  {$3$};
\node[exnode] (v2) at (0, 2)  {$2$};
\node[exnode] (v4) at (2, 2)  {$4$};

\node[exnode] (v5) at (4, 2)  {$5$};
\node[exnode] (v6) at (5.2, 2)  {$6$};
\node[exnode] (v7) at (6, 3)  {$7$};
\node[exnode] (v8) at (6, 1)  {$8$};

\draw[-, exedge, bend left = 0] (v1) to (v2);
\node (v1v2) at (-0.3,2.5) {$1/2$};

\draw[-, exedge, bend left = 0] (v3) to (v2);
\node (v3v2) at (-0.3,1.5) {$1/2$};

\draw[-, exedge, bend left = 0] (v2) to (v4);
\node (v4v2) at (1,2.15) {$1$};

\draw[-, exedge, bend left = 0] (v1) to (v4);
\node (v1v4) at (1.15,2.65) {$1/2$};

\draw[-, exedge, bend left = 0] (v3) to (v4);
\node (v3v4) at (1.15,1.35) {$1/2$};

\draw[-, exedge, bend left = 0] (v4) to (v5);
\node (v4v5) at (3,2.15) {$0$};

\draw[-, exedge, bend left = 0] (v5) to (v6);
\node (v5v6) at (4.7,2.15) {$1$};

\draw[-, exedge, bend left = 0] (v5) to (v7);
\node (v5v7) at (4.9,2.65) {$1$};

\draw[-, exedge, bend left = 0] (v5) to (v8);
\node (v5v8) at (4.9,1.35) {$1$};

\draw[-, exedge2, bend left = 0] (v6) to (v7);
\node (v6v7) at (5.7,2.45) {$2$};

\draw[-, exedge2, bend left = 0] (v6) to (v8);
\node (v6v8) at (5.7,1.55) {$2$};

\draw[-, exedge2, bend left = 0] (v7) to (v8);
\node (v7v8) at (6.15,2) {$2$};
\end{tikzpicture}}
\caption{\ref{eq:LP2-STC} (d=1).}
\label{fig:toylp2}
\end{subfigure}

\begin{subfigure}[b]{0.4\textwidth}
\centering
\resizebox{\linewidth}{!}{
\begin{tikzpicture}
\tikzstyle{exedge} = [yafcolor5!80, thick, text=black!80]
\tikzstyle{exedge2} = [yafcolor2!80, very thick, text=black!80]
\tikzstyle{exedge_epsi} = [yafcolor4!80, thick, dashed, text=black!80]
\tikzstyle{exnode} = [thick, draw = yafcolor7!80, fill=white, circle, inner sep = 1pt, text=black, minimum width=11pt]

\node[exnode] (v1) at (0, 3)  {$1$};
\node[exnode] (v3) at (0, 1)  {$3$};
\node[exnode] (v2) at (0, 2)  {$2$};
\node[exnode] (v4) at (2, 2)  {$4$};

\node[exnode] (v5) at (4, 2)  {$5$};
\node[exnode] (v6) at (5.2, 2)  {$6$};
\node[exnode] (v7) at (6, 3)  {$7$};
\node[exnode] (v8) at (6, 1)  {$8$};

\draw[-, exedge, bend left = 0] (v1) to (v2);
\node (v1v2) at (0.2,2.5) {$2$};

\draw[-, exedge, bend left = 0] (v3) to (v2);
\node (v3v2) at (0.2,1.5) {$2$};

\draw[-, exedge, bend left = 0] (v2) to (v4);
\node (v4v2) at (1,2.15) {$2$};

\draw[-, exedge, bend left = 0] (v1) to (v4);
\node (v1v4) at (1.15,2.65) {$2$};

\draw[-, exedge, bend left = 0] (v3) to (v4);
\node (v3v4) at (1.15,1.35) {$2$};

\draw[-, exedge, bend left = 0] (v4) to (v5);
\node (v4v5) at (3,2.15) {$-1$};

\draw[-, exedge, bend left = 0] (v5) to (v6);
\node (v5v6) at (4.7,2.15) {$2$};

\draw[-, exedge, bend left = 0] (v5) to (v7);
\node (v5v7) at (4.9,2.65) {$2$};

\draw[-, exedge, bend left = 0] (v5) to (v8);
\node (v5v8) at (4.9,1.35) {$2$};

\draw[-, exedge2, bend left = 0] (v6) to (v7);
\node (v6v7) at (5.7,2.45) {$2$};

\draw[-, exedge2, bend left = 0] (v6) to (v8);
\node (v6v8) at (5.7,1.55) {$2$};

\draw[-, exedge2, bend left = 0] (v7) to (v8);
\node (v7v8) at (6.15,2) {$2$};

\draw[-, exedge_epsi, bend left = -50] (v1) to (v3);
\node (v1v3) at (-0.7,2) {$2$};

\end{tikzpicture}}
\caption{\ref{eq:LP4-STC} (d=1, C=1).}
\label{fig:toylp4}
\end{subfigure}
\caption{\label{fig:toyexample}Toy example with 8 nodes to show the different outcomes of the proposed algorithms. The triangle edges are shown in orange.}
\end{figure}
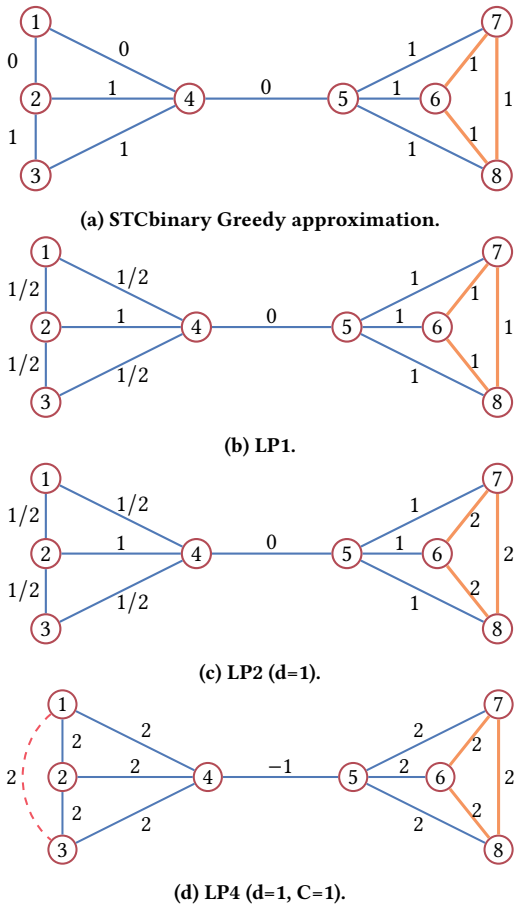
To gain some insight in our methods, we start by discussing a simple toy example. Figure~\ref{fig:toyexample} shows a network of 8 nodes, modelling a scenario of 2 communities being connected by a bridge, i.e., the edge $\{4,5\}$.
The nodes $\{1,2,3,4\}$ form a near-clique---the edge $\{1,3\}$ is missing---while the nodes $\{5,6,7,8\}$ form a 4-clique. This 4-clique contains a triangle clique: the subgraph induced by the nodes $\{6,7,8\}$.
Triangle edges are colored orange in the figure.

Fig.~\ref{fig:toybin} contains a solution to \STCbinary.
Fig.~\ref{fig:toylp1} shows a half-integral optimal solution to Problem \ref{eq:LP1-STC}. 
We observe that for \STCbinary\ we could swap nodes 1 and 3 and obtain a different yet equally good solution, hence the strength assignment is arbitrary with respect to several edges, while for \ref{eq:LP1-STC} the is not the case. Indeed, there is no evidence to prefer a strong label for edges $\{2,3\}$ and $\{3,4\}$ over the edges $\{1,2\}$ and $\{1,4\}$.

Figure~\ref{fig:toylp2} shows a symmetric optimal solution to Problem~\ref{eq:LP2-STC}, allowing for multi-level edge strengths. It labels the triangle edges as stronger than all other (wedge) edges, in accordance with Theorem~\ref{thm:all} and \ifkdd\else Corollary~\ref{cor:half-integrality}\fi.

Finally, Figure~\ref{fig:toylp4} shows the outcome of \ref{eq:LP4-STC} for $d=1$ and $C=1$, allowing for edge additions and deletions.
For $C=0$, the problem becomes unbounded: the edge $\{4,5\}$ is only part of wedges, and since wedge violations are unpenalized, $w_{45}=+\infty$ is the best solution (see Section~\ref{sec:negedge}).
Since this edge is part of 6 wedges, the problem becomes bounded for $C>1/6$.
For $C=1$, the algorithm produces a value of 2 for the absent edge $\{1,3\}$.
This suggests the addition of an edge $\{1,3\}$ with strength $2$ to the network, in order to increase the objective function.
\ifkdd\else
This is the only edge being suggested for addition by the algorithm.
\fi
Edge $\{4,5\}$, on the other hand, is given a value of $-1$.
As discussed in Section~\ref{sec:negedge}, this corresponds to the strength of an
absent edge (when $d=1$), suggesting the removal of the bridge in the network in
order to increase the objective.

\note[Jef->Florian]{I do not understand the comment about $\{4,5\}$ only being in wedges hence the problem becoming unbounded. Both existing and non-existing edges are lower-bounded so how can $\{4,5\}$ be bounded by any constraints anymore?}

\ifkdd\else
For large $C$ there will be no more edge additions being suggested, as can be seen by setting $C=\infty$ in LP4-STC (reducing it to LP3-STC).
The cost of a violation of a wedge constraint will always be higher than the possible benefits.
However, regardless of the value of $C$, the edge $\{4,5\}$ is always being suggested for edge deletion.
\fi

A further illustration on a more realistic network is given in Fig~\ref{fig:heatmap-lesmis},
which shows the edge strengths assigned by \STCbinary\ (1st), \ref{eq:LP1-STC} (2nd), \ref{eq:LP2-STC} with $d=1$ (3rd), and \ref{eq:LP4-STC} with $d=1$ and $C=1$ (4th).
Also here, we see that \STCbinary\ is forced to make arbitrary choices,
while \ref{eq:LP1-STC}, and \ref{eq:LP2-STC} avoids this by making use of an intermediate level.
Densely-connected parts of the graph tend to contain edges marked as strong,
with an extra level of strength for \ref{eq:LP2-STC} assigned to the triangle edges.
In comparison with~\ref{eq:LP2-STC},~\ref{eq:LP4-STC} suggests to remove a lot of
weak edges (weight 0 in \ref{eq:LP2-STC})
that act as bridges between the communities,
in order to allow a stronger labeling in the densely-connected regions.
Besides edge removal, it also suggests the addition of edges in a near-cliques to form full cliques.

\note[Jef->Florian]{How do we see that \STCbinary\ is forced to make arbitrary choices?}


\begin{figure}[tp]
    \centering
\includegraphics[width=\columnwidth]{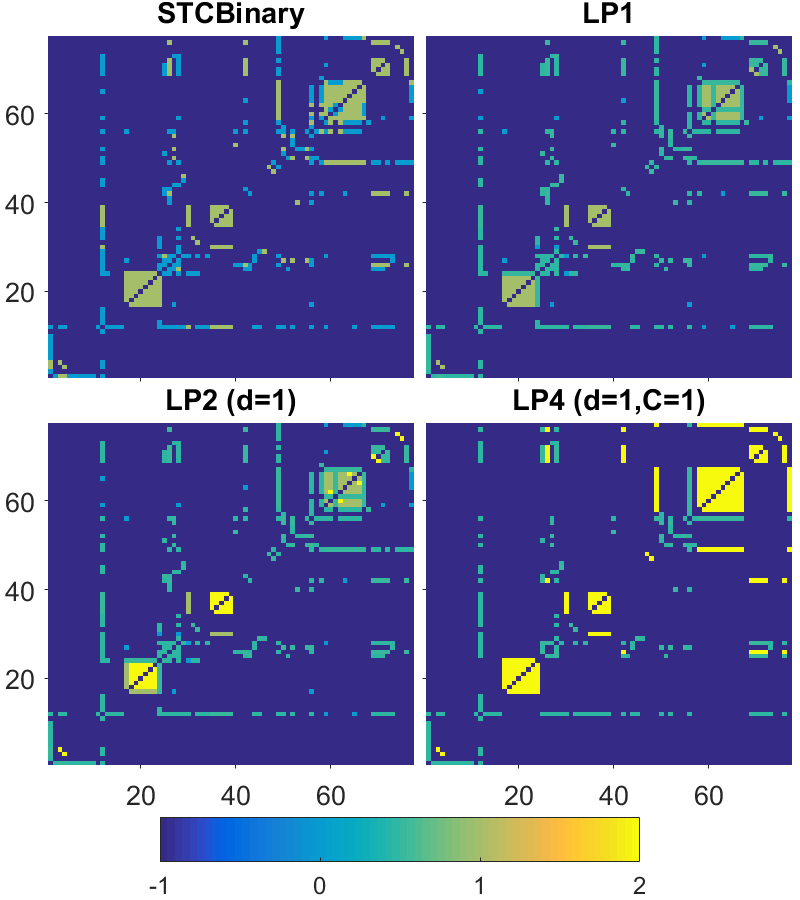}
    \caption{Heatmap of the edge strengths on Les Miserables with methods~\STCbinary\ Greedy (1st),~\ref{eq:LP1-STC} (2nd),~\ref{eq:LP2-STC} with $d=1$ (3rd) and~\ref{eq:LP4-STC} with $d=1$ and $C=1$ (4th). A strength of $-1$ indicates there is no edge.}
		\label{fig:heatmap-lesmis}
\end{figure}

\subsection{Objective performance analysis}

\note[Tijl]{I shortened this title to make some space (and also because we might forget to adapt it if we end up not doing link prediction!}

We evaluate our approaches in a similar manner as Sintos and Tsaparas~\cite{sintos2014using}.
In particular, we investigate whether the optimal strength assignments correlate to externally provided ground truth measures of tie strength,
on a number of networks for which such information is available.
Table~\ref{table:stats} shows a summary of the dataset statistics and edge weight interpretations. 

\begin{table}
\caption{Network statistics.\label{table:stats}}
\resizebox{\columnwidth}{!}{%
\begin{tabular}{lrrl}
\toprule
\multicolumn{1}{l}{Network}&
\multicolumn{1}{c}{Vertices}&
\multicolumn{1}{c}{Edges}&
\multicolumn{1}{c}{Edge weight meaning}\\
\midrule
Les Mis.&77&254&{co-appearence of characters in same chapter}\\
KDD&2\,738&11\,073&{co-authorship between 2 authors}\\
Facebook \cite{viswanath-2009-activity}&3\,228&4\,585&{number of posts on each other's wall}\\
Twitter&4\,185&5\,680&{mentions of each other}\\
Authors&9\,150&34\,614&{unknown}\\
BitCoin OTC&5\,875&21\,489&{Who-trust-whom score in Bitcoin OTC}\\
BitCoin Alpha&3\,775&14\,120&{Who-trust-whom score in Bitcoin Alpha}\\
\bottomrule
\end{tabular}%
}
\end{table}

We compare the algorithms \STCbinary\ Greedy 
(which Sintos and Tsaparas found to perform best), \ref{eq:LP1-STC} and \ref{eq:LP2-STC}.
For each dataset, the first row in Table~\ref{table:weights} 
displays the number of edges that are assigned in that category.
The second row shows the mean ground truth weight over the labeling assigned by the respective algorithm.

\note[Jef]{I think reviewers will find it very unconvincing that we did not test LP4 unless we present very good reasons to exclude it. If we do not want to have it in the table, we should just outright admit it performs very poorly.}
\note[Tijl->Florian, all]{I agree. Note that the section title still said we do link prediction too, so I removed that for now.
My preference would be that we evaluate LP4 for link prediction only, perhaps without comparison, in a very short paragraph.
In the tech report we could elaborate.
Of course, LP4 is also just a different method in that it assigns strengths to non-existent edges,
so we could also argue that a simple comparison cannot be made.
That would be the quick and easy solution if no time remains.}

Les Miserables is a network where \STCbinary\ Greedy is known to perform well~\cite{sintos2014using}. 
For this dataset, we can clearly see that our methods provide a correct multi-level strength labeling, 
enabling more refined notions of tie strength. 

A second observation is that in for the networks 
KDD, Facebook, Twitter, and Authors, neither the existing nor the newly-proposed methods perform well.
This raises the question of whether the STC assumption is valid in these networks 
with the provided ground truths.%
\footnote{For example, in a co-authorship network,
junior researchers having published their first paper with several co-authors
could well have all their first edges marked as strong,
as their co-authors are connected through the same publication.
Yet, they have not yet had the time to form strong connections according to the ground truth.
Also, although the Facebook and Twitter networks are social networks,
and hence have a natural tendency to satisfy the STC property \cite{David:2010:NCM:1805895},
these sampled networks are too sparse to accommodate a meaningful number of strong edges without any STC violations.}
That said, it is reassuring to see that our methods work in a robust and fail-safe way:
in such cases, as indicated by the high number of 1/2 strength assignments.

For trust networks in particular, however,
it has been described that the STC property is likely to develop due to the transitive property~\cite{David:2010:NCM:1805895}.
Indeed, if a user A trusts user B and user B trusts user C, then user A has a basis for trusting user C.
The two BitCoin networks are examples of such trust networks.
Our methods perform well in identifying some clearly strong and some clearly weak edges,
although it again takes a cautious approach in assigning an intermediate strength to many edges.
Remarkably though, \STCbinary\ Greedy performs poorly on this network,
incorrectly labeling many strong edges as weak and vice versa.

\begin{table}
\caption{Mean ground-truth weight analysis comparison of different STC methods.
For each dataset, the first row is the number of edges of an assigned label.
The 2nd row indicates the mean groundtruth weight over that respective set of edges.
The ground-truth strength ranges are indicated by the numbers between brackets. \label{table:weights}}
\resizebox{\columnwidth}{!}{%
\begin{tabular}{lrrcrrrcrrrr}
\toprule
\multicolumn{1}{l}{Network}&
\multicolumn{2}{c}{\STCbinary\ Greedy} & $\quad$ &
\multicolumn{3}{c}{\ref{eq:LP1-STC}} & $\quad$ &
\multicolumn{4}{c}{\ref{eq:LP2-STC} (d=1)}\\
\cmidrule{2-3}
\cmidrule{5-7}
\cmidrule{9-12}
&
\multicolumn{1}{c}{1}& 
\multicolumn{1}{c}{0}& &
\multicolumn{1}{c}{1}&\multicolumn{1}{c}{$1/2$}&\multicolumn{1}{c}{0}& &
\multicolumn{1}{c}{2}&\multicolumn{1}{c}{1}&\multicolumn{1}{c}{$1/2$}&\multicolumn{1}{c}{0}\\
\midrule
Les Mis. &131&123&&60&180&14&&30&30&180&14\\
$[1$--$31]$ &3.6&2.8&&4.5&2.9&1.5&&3.3&5.7&2.9&1.5\\
\midrule
KDD&3\,085&7\,988&&545&10\,390&138&&290&252&10\,396&135\\
$[0.04$--$47.3]$ &1.14&0.85&&0.89&0.93&0.61&&0.77&1.03&0.94&0.61\\
\midrule
Facebook&1\,451&3\,134&&28&4\,547&10&&11&17&4\,547&10\\
$[1$--$30]$ &1.9&1.94&&2.29&1.92&1.5&&2.46&2.18&1.92&1.5\\
\midrule
Twitter&282&5\,398&&0&5\,680&0&&0&0&5\,680&0\\
$[1$--$139]$ &1.29&2&&-&1.97&-&&-&-&1.97&-\\
\midrule
Authors&16\,647&17\,967&&9\,599&22\,994&2\,021&&5\,590&4\,009&22\,994&2\,021\\
$[1$--$52]$ &1.19&1.4&&1.1&1.41&1.16&&1.09&1.1&1.41&1.16\\
\midrule
BitCoin OTC&1\,794&19\,695&&37&21\,446&6&&26&11&21\,446&6\\
$[-10$--$10]$ &0.89&0.62&&2.37&0.64&-2.33&&2.5&2.1&0.64&-2.33\\
\midrule
BitCoin Alpha&1\,178&12\,942&&6&14\,113&1&&4&2&14\,113&1\\
$[-10$--$10]$ &1.21&1.43&&5&1.4&-10&&6&3&1.4&-10\\
\bottomrule
\end{tabular}%
}
\end{table}


\oldnote[Tijl]{Running time analysis at least of the LP2 formulation
solved using CVX
as well as solved using the Hochbaum implementation.
Perhaps also for the LP3 and LP4 formulations just using CVX.}

Finally, Table~\ref{tab:runtime} reports running times
on a PC with an Intel i7-4800MQ CPU at 2.70GHz and 16 GB RAM
of our CVX/MOSEK and Hochbaum-Naor implementations.
It demonstrates the superior performance of the latter.
Remarkably, the Hochbaum-Naor algorithm performs very comparably to the greedy approximation algorithm for STCbinary.

\begin{table}
\caption{\label{tab:runtime}Running times (in seconds) 
for solving \ref{eq:LP2-STC-sym} with $d=1$ by a general LP solver with Interior Point (IP) algorithm, 
for Hochbaum-Naor's algorithm with Minimum Cut as the main subroutine, and 
for the greedy approximation algorithm for STCbinary. 
Total times include problem construction for \ref{eq:LP2-STC-sym}, 
graph construction for Hochbaum-Naor, and wedge-graph construction for STCbinary.}	
\resizebox{\columnwidth}{!}{%
\begin{tabular}{lrrcrrcrr}
\toprule
\multicolumn{1}{l}{Network} &
\multicolumn{2}{c}{LP}& $\quad$ &
\multicolumn{2}{c}{Hochbaum}& $\quad$ &
\multicolumn{2}{c}{\STCbinary}\\
\cmidrule{2-3}
\cmidrule{5-6}
\cmidrule{8-9}
& \multicolumn{1}{c}{IP} & \multicolumn{1}{c}{total} && 
\multicolumn{1}{c}{MinCut} & \multicolumn{1}{c}{total} && 
\multicolumn{1}{c}{Greedy} & \multicolumn{1}{c}{total}\\
\midrule
Les Mis.& 0.11& 0.63 && 0.004 & 0.008 && 0.002 & 0.016\\
KDD & 3.92 &19.30 && 0.29 & 1.60 && 0.86 & 2.07 \\
Facebook & 0.31 & 1.94 &&  0.02 & 0.46 && 0.15 & 0.41\\
Twitter & 1.44 & 10.30 && 0.28 & 0.87 && 0.27 & 1.88\\
Authors & 14.66 & 47.29 && 0.68 & 6.68 && 5.22 & 9.50\\
BitCoin OTC & 126.22 & 269.31 && 8.69 & 13.41 && 1.54& 7.12\\
\bottomrule
\end{tabular}%
}
\end{table}

\oldnote[Jef]{Again, should we not just report the times for STCmin Greedy as well? We do not have to hide how fast it is, it is just an approximation algorithm to a clearly inferior problem setting.}
\oldnote[Tijl->Polina/Florian]{I agree, it would be better (even if it's our own implementation),
especially because there is space for another column. Polina, could you still do that?}

\section{Related work\label{sec:relatedwork}}

\ifkdd\else
This paper builds on the STC principle,
which was proposed in sociology by Simmel \cite{sim:08}.
Sintos and Tsaparas \cite{sintos2014using} first considered the problem of labeling the edges of the graph to enforce the STC property:
maximize the number of strong edges, such that the network satisfies the STC property.
In our work we relax and extend this formulation by introducing new constraints and integer labels.
To our knowledge, we are the first to introduce and study such formulations.
\fi

The 
work by Sintos and Tsaparas~\cite{sintos2014using} is part of a broader line of active recent research
aiming to infer the strength of the links in a social network.
E.g., Jones et al.~\cite{jones2013inferring} uses frequency of online interaction 
to predict of strength ties with high accuracy.
Gilbert et al.~\cite{gilbert2009predicting} characterize social ties based on similarity and interaction information.
Similarly, Xiang et al.~\cite{xiang2010modeling} estimate relationship strength from homophily principle and interaction patterns
and extend the approach to heterogeneous types of relationships.
Pham et al.~\cite{pham2016inferring} incorporate spatio-temporal features of social iterations to increase accuracy of inferred tie strength. 

A related direction of research focuses solely on inferring types of the links in a network.
E.g., Tang et al.~\cite{tang2012inferring,tang2011learning,tang2016transfer} propose a generative statistical model,
which can be used to classify heterogeneous relationships.
The model relies on social theories and incorporates structural properties of the network and node attributes.
Their more recent works can also produce strength of the predicted types of ties.
Backstrom et al.~\cite{backstrom2014romantic} focuses on the graph structure to identify a particular type of ties---romantic relationships in Facebook.

Most of these works, however, make use of various meta-data and characteristics of social interactions in the networks.
In contrast, like Sintos and Tsaparas' work,
our aim is to infer strength of ties solely based of graph structure, and in particular on the STC assumption.

Another recent extension of the work of Sintos and Tsaparas \cite{sintos2014using} is 
followed by Rozenshtein et al.~\cite{rozenshtein2017inferring}.
However, their direction is different:
they consider binary strong and weak labeling with additional community 
connectivity constrains and allow STC violations to satisfy those constraints.

\section{Conclusions and further work}\label{sec:conclusions}

\ifkdd\else
\subsection{Conclusions}
\fi

We have proposed a sequence of linear programming relaxations of the STCbinary problem 
introduced by Sintos and Tsaparas~\cite{sintos2014using}.
These formulations have a number of advantages,
most notably their computational complexity,
the fact that they refrain from making arbitrary strength assignments in the presence of uncertainty,
and as a result, enhanced robustness.
Extensive theoretical analysis of the second relaxation (\ref{eq:LP2-STC})
has not only provided insight into the solution and the arbitrariness the solution from STCbinary may exhibit,
it also yielded a highly efficient algorithm for finding a symmetric (non-arbitrary) optimal strenght assignment.

The empirical results confirm these findings.
At the same time, they raise doubts about the validity of the STC property in real-life networks,
with trust networks appearing to be a notable exception.

\ifkdd\else
\subsection{Further work}
\fi

Our research results open up a large number of avenues for further research.
The first is to investigate whether more efficient algorithms can be found
for inferring the range of edge strengths across the optimal face of the feasible polytope.
A related research question is whether the marginal distribution of individual edge strengths,
under the uniform distribution of the optimal polytope,
can be characterized in a more analytical manner (instead of by uniform sampling).
Both these questions seem important beyond the STC problem,
and we are unaware of a definite solution to them.

A second line of research is to investigate alternative problem formulations.
An obvious variation would be to take into account community structure,
and the fact that the STC property probably often fails to hold for wedges that span different communities.
A trivial approach would be to simply remove the constraints for such wedges,
but more sophisticated approaches could exist.
Additionally, 
it would be interesting to investigate the possibility 
to allow for different relationship types and respective edge strengths,
requiring the STC property to hold only within each type.
Furthermore, the fact that the presented formulations are LPs,
combined with the fact that many graph-theoretical properties can be expressed in terms of linear constraints,
opens up the possibility to impose additional constraints on the optimal strength assignments
without incurring significant computational overhead as compared to the interior point implementation.
One line of thought is to impose upper bounds on the sum of edge strengths incident to any given edge,
modeling the well-known fact that an individual is limited in how many strong social ties they can maintain.

A third line of research is whether an active learning approach can be developed,
to quickly reduce the number of edges assigned an intermediate strength by our approaches.

\ifkdd\else
More directly, a fourth line of research concerns the question of whether the theoretical understanding of
Problems~\ref{eq:LP1-STC} and \ref{eq:LP2-STC} can be transferred more completely to
Problems~\ref{eq:LP3-STC} and \ref{eq:LP4-STC} than achieved in the current paper.
\fi

Finally, perhaps the most important line of further research concerns the validity of the STC property:
could it be modified so as to become more widely applicable across real-life social networks?

\note[Tijl]{In the KDD version we could further condense the list of further work ideas if we need more space.}

\ptitle{Acknowledgements} This work was supported by
the ERC under the
EU's Seventh Framework Programme (FP7/2007-2013) / ERC
Grant Agreement no.\ 615517,
FWO (project no.\ G091017N, G0F9816N),
the EU's Horizon 2020 research and innovation programme and the FWO under the Marie Sklodowska-Curie Grant Agreement no.\ 665501,
three Academy of Finland projects  (286211, 313927, and 317085), 
and the EC H2020 RIA project ``SoBigData'' (654024).

\note[Aris]{Comment out acks for the submission, if need space.}

\note[Jef]{Should we update the title of the arXiv report?}
\note[Tijl]{As it is pretty much the same paper just with more text, we could probably keep the same title?}

\bibliographystyle{ACM-Reference-Format}
\bibliography{paper}

\end{document}